\documentclass[pra]{revtex4-2} 
\usepackage{amssymb,amsmath,bm} 
\usepackage{mathrsfs}
\usepackage{cases}
\usepackage{graphicx}

\makeatletter

\def \V#1{{\bm #1}}
\def \be {\begin{equation}}
\def \ee {\end{equation}}
\def \bml{\begin{multline}}
\def \eml{\end{multline}}
\newcommand{\Exp}[1]{\,\mathrm{e}^{\mbox{\footnotesize$#1$}}}
\newcommand{\I}{\mathrm{i}}

\newcommand{\cket}[1]{|#1)}
\newcommand{\cbra}[1]{(#1|}
\newcommand{\tr}[1]{\mathrm{tr}\{#1\}}
\newcommand{\Tr}[1]{\mathrm{Tr}\Big\{#1\Big\}}
\newcommand{\Det}[1]{\mathrm{Det}\{#1\}}

\def \del{\partial}

\renewcommand{\openone}{\mathbb{I}}

\def \cA{{\cal A}}

\def \sofctwo{{\cal S}({\mathbb C}^2)}
\def \cB{{\cal B}}
\def \cE{{\cal E}}
\def \cJ{{\cal J}}

\def \PsetX{{\cal P}(\cX)}

\def \cH{{\cal H}}
\def \cP{{\cal P}}
\def \cS{{\cal S}}

\def \cX{{\cal X}}

\def \sofh{{\cal S}({\cal H})}

\def \bbr{{\mathbb R}}
\def \bbc{{\mathbb C}}

\def \bbn{{\mathbb N}}
\def \ds{\displaystyle}
\def \e{\mathsf e}
\def \p{\mathsf p}
\def \u{\mathsf u}

\def \Lra{\Leftrightarrow}
\def \sld{\mathrm{SLD}}

\newtheorem{theorem}{Theorem}[section]
\newtheorem{lemma}[theorem]{Lemma}
\newtheorem{proposition}[theorem]{Proposition}
\newtheorem{corollary}[theorem]{Corollary}

\newenvironment{proof}[1][Proof:]{\begin{trivlist}
\item[\hskip \labelsep {\bfseries #1}]}{\end{trivlist}}

\newcommand{\qed}{\nobreak \ifvmode \relax \else
      \ifdim\lastskip<1.5em \hskip-\lastskip
      \hskip1.5em plus0em minus0.5em \fi \nobreak
      \vrule height0.75em width0.5em depth0.25em\fi}
\makeatother

\begin{document}
%
%


\title{Quantum-state estimation problem via optimal design of experiments}
\author{Jun Suzuki}
\date{\today}
\email{junsuzuki@uec.ac.jp}
\affiliation{
Graduate School of Informatics and Engineering, The University of Electro-Communications,\\
1-5-1 Chofugaoka, Chofu-shi, Tokyo, 182-8585 Japan
}

\begin{abstract}
In this paper, we study the quantum-state estimation problem in the framework of 
optimal design of experiments. 
We first find the optimal designs about arbitrary qubit models for popular optimality criteria such as $A$-, $D$-, and $E$-optimal designs. 
We also give the one-parameter family of optimality criteria which includes these criteria. 
We then extend a classical result in the design problem, 
the Kiefer-Wolfowitz theorem, to a qubit system showing the $D$-optimal design is equivalent to a certain type of the $A$-optimal design. 
We next compare and analyze several optimal designs based on the efficiency. 
We explicitly demonstrate that an optimal design for a certain criterion can be highly inefficient for other optimality criteria. 
\end{abstract}


\maketitle

\section{Introduction}\label{sec:intro}
Studies on any experiment governed by the law of statistics consists of two different stages. 
The first step is to design or prepare a good experimental setup to extract information of interest. 
The second one is to analyze an actual datum obtained from the chosen experiment. 
The standard textbooks on statistics focus only on the latter part; 
that is, how to extract a quantity of interest from a given datum. 
The first element, known as optimal design of experiments (DoE), is a well-established branch of  classical statistics.\cite{fisher1960design,fedorov,pukelsheim} 
It provides a systematic and powerful tool to search for an optimal DoE under a given optimality criterion. 

Quantum state estimation problems\cite{helstrom,holevo,QSEbook,hayashi2016quantum,petz,teo2016} are naturally divided into two stages, since measurement outcomes obey the statistical rule by the axioms of quantum theory. It seems, however, that textbooks on the subject do not emphasize this point clearly nor analyze the problem at hand in the language of optimal DoE. 
Although several authors applied this theory to a quantum system,
\cite{kosut2004optimal,nunn2010optimal,ballo2012convex,ruppert2012optimal,stm12} 
its use has been limited so far. 
One important message of this paper is that so called ``incompatibility" of estimating two different parameters is already well-known phenomena in the classical theory of optimal DoE. 
Therefore, we cannot immediately attribute this kind of trade-off relations to quantum nature of the problem. 

In a recent paper~\cite{gns19}, we developed the general theory to estimate a family of quantum channels based on the theory of optimal DoE. 
We made an explicit comment there that quantum-state estimation problems can be handled as a special case. 
The aim of this paper is two folds. 
First, we provide a framework of optimal DoE for the problem of quantum state estimation. 
We then apply the standard methodology of characterizing an optimal design. 
We show that the qubit case is completely solved as an optimal DoE problem. 
Second, we wish to compare different optimal designs for a qubit system. 
Thereby, we explicitly demonstrate that a particular optimal design is not optimal for others. 

In the classical theory of DoE, it seems that a systematic comparison of different optimal designs is not a common subject. Rather, they focus on analyzing a problem at hand based on a particularly chosen optimality criterion. 
In this study, we emphasize that a proper comparison among different optimality criteria is necessary rather than adopting one particular optimality criterion. 
This is because there is no universally accepted optimality criterion exists, but it is very subjective. 
Second reason is that one particular optimal design may become inefficient for the other optimality criteria. 
Indeed, our study suggests that one of the common optimal criteria, the $D$-optimality, in classical statistics may not be suited for quantum-state estimation problems. 
This is based on the result of the general qubit model in the tomographic scenario. 
Other optimal designs are shown to perform very poorly for the $D$-optimal criterion when the purity of quantum states is high. 

The outline of this paper is as follows. 
In Sec.~\ref{sec2}, a brief summary on the classical theory of optimal DoE is given. 
We then apply it to the problem on quantum-state estimation in Sec.~\ref{sec3}. 
We also analytically solve common optimal design for the general qubit system. 
In Sec.~\ref{sec4}, we derive a quantum version of the equivalent theorem in the qubit system. 
Section \ref{sec5} studies comparisons of different optimality criteria. 
We close our paper by conclusions and remarks in Sec.~\ref{sec6}. 

\section{Preliminaries}\label{sec2}
In this section, we provide a brief summary of optimal DoE base on non-linear response theory\cite{fedorov,fh97,pp13,fl14} 
Our formulation is based on the result presented in Ref.~\cite{gns19}. 

\subsection{Formulation}\label{sec:formulation}
\subsubsection{Terminologies and definitions}
Suppose a physical system of interest is specified by a state $s$, 
and we denote a set of all possible states by $\cS$. 
We call the set $\cS$ as the state space. 
Typically, $\cS$ is a subspace of a vector space, which could be real or complex in general. 
Denote by ${\bm \theta}=(\theta_1,\theta_2,\dots,\theta_n)$ an $n$-parameter coordinate system 
for the state space, called a {\it model parameter}, or simply {\it parameter}, to describe a state by $s_{\bm \theta}$. 
Our interest is to analyze a family of states $\{s_{\bm \theta}\,|\, {\bm \theta}\in\Theta\}$,  
where the parameter ${\bm \theta}=(\theta_i)$ takes values in $\Theta$, 
which is an open subset of $\bbr^n$. 
In the following, we assume that ${\bm \theta}\mapsto s_{\bm \theta}$ is one-to-one and smooth in $\theta$. 
Therefore, the model parameter ${\bm \theta}$ identifies the state $s_{\bm \theta}$ uniquely. 
A {\it design} $\e$ describes a particular experimental setup, 
and $\cE$ denotes the set of all possible designs, which will be called a {\it design space}. 
A {\it model function} $f$ is a mapping from $\cS\times\cE$ to a set of probability distributions on $\cX (=: \PsetX$).  
That is, $f:\, (s,\e)\mapsto p_s(\cdot|\e)\in\PsetX$ where $\forall x\in\cX$, $p_s(x|\e)\ge0$ and 
$\sum_{x\in\cX}p_s(x|\e)=1$ hold due to the axiom of probability theory. 
A familiar example of this kind is a linear regression model that has been intensively studied in the field of 
optimal DoE.\cite{pukelsheim,fh97,pp13,fl14} 

One of the main objectives of optimal DoE is to infer the unknown parameter ${\bm \theta}$ 
and hence $s_{\bm \theta}\in\cS$ by choosing an appropriate design $\e$. 
A difference from the usual setting of classical statistics is that 
a statistical model is specified by the conditional distribution according to a chosen design $\e$:  
\begin{equation*}
M(\e)= \{p_{\bm \theta}(\cdot|\e)\,|\, {\bm \theta}\in\Theta\},
\end{equation*} 
where $p_{\bm \theta}(\cdot|\e)=p_{s_{\bm \theta}}(\cdot|\e)$ is a shorthand convention. 
And thus, a datum $X$ is a random variable drawn according to $p_{\bm \theta}(\cdot|\e)$, which depends 
on a particular choice of design $\e$. 
Denoting the conditional expectation value by  $E_{\bm \theta}[X|\e]:= \sum_{x\in\cX} x\, p_{\bm \theta}(x|\e) $, 
the mean-square error (MSE) matrix for the estimator $\hat{{\bm \theta}}:\cX\to\Theta$ is defined by
\begin{equation*}
V_{\bm \theta}[\hat{{\bm \theta}}|\e]:=\Big[ E_{\bm \theta}\big[(\hat{\theta}_i(X)-\theta_i) (\hat{\theta}_j(X)-\theta_j)\big|\e\big] \Big]_{i,j}.
\end{equation*}
We look for the best estimator and design under the condition of 
locally unbiasedness defined as follows. 
An estimator $\hat{{\bm \theta}}=(\hat{\theta}_i)$ is said to be {\it locally unbiased at ${\bm \theta}$} under a design $\e$, 
if $E_{\bm \theta}[\hat{\theta}_i(X)|\e]=\theta_i$ and $\frac{\del}{\del\theta_j}E_{\bm \theta}[\hat{\theta}_i(X)|\e]=\delta_{ij}$ 
are satisfied for $\forall i,j$ at ${\bm \theta}$. 
When an estimator is locally unbiased at all points ${\bm \theta}$, it is an unbiased estimator. 
As explained later, locally unbiasedness for $\hat{{\bm \theta}}$ under the design $\e$ is 
fundamental in the theory of DoE. 
This is in contrast to the standard parameter estimation problem, where locally unbiased estimators 
are of no practical importance in general.

For a fixed design $\e\in\cE$, and assume that the model $M(\e)$ satisfies a certain regularity conditions. 
We can then apply the Cram\'er-Rao (CR) inequality, 
\begin{equation*}
V_{\bm \theta}[\hat{{\bm \theta}}|\e]\ge \Big( J_{\bm \theta}[\e]\Big)^{-1}, 
\end{equation*}
which holds for any locally unbiased estimator. 
Here $J_{\bm \theta}[\e]$ is the {\it Fisher information matrix} about the statistical model $M(\e)$ at ${\bm \theta}$, 
which is defined by 
\begin{equation*}
J_{\bm \theta}[\e]:= \Big[ E_{\bm \theta}\big[ \frac{\del \ell_{\bm \theta}(X|\e)}{\del\theta_i}  \frac{\del\ell_{\bm \theta}(X|\e)}{\del\theta_j} \big|\e \big] \Big]_{i,j}, 
\end{equation*}
where $\ell_{\bm \theta}(x|\e):=\log p_{\bm \theta}(x|\e) $ is the logarithmic likelihood function. 
Note that a locally unbiased estimator always exists at each point, 
and thus the right hand side of the CR inequality provides the fundamental limit for the MSE matrix 
at a given point. 
With this fact, we aim to minimize the inverse of the Fisher information matrix over all possible designs. 

In passing we make remarks on the figure of merit formulated in this paper and other formulations. 
It is also possible to minimize other quantities related to estimation errors, such as the fidelity, trace distance, and quantum relative entropy between an estimated state and the true state. 
The current paper focuses on the design problem before actual experiments are performed. 
Thus, it is more natural to maximize information about an unknown state as possible, which is measured by the classical Fisher information. 
In particular, we are interested in the best experimental design at each point. 
This is to say that an optimal design here is optimal locally.  
By contrast, we can also investigate optimal designs on average over possible states with some prior distribution on the state space. 
This is called the Bayesian design problem.\cite{pukelsheim,chaloner1995bayesian,dasgupta199629,ryan2016review}  

Another formulation is to find the optimal design for the worst case. 
This is known as the min-max design problem. 
Extension of the present work to these setting shall be presented elsewhere. 
Finally, a recent paper \cite{lu2020generalized} studied a family of precision bounds for a function of the MSE matrix using the concept of weighted $f$-mean, which is known in the positive matrix theory. 
While we have similar optimization problems, the current paper is based on the standard statistical tool rather than purely mathematical ones. 

\subsubsection{Optimality criteria}
To proceed further, we consider different types of optimal designs 
defined by each optimality function. 
Let $\Psi$ be a real-valued function of non-negative matrix, called an {\it optimality function}, 
such that $\Psi(A)\ge0 $ for all $A\ge0$. We can then formulate our optimization problem as follows \footnote{In this paper, we use loose language in mathematics. When the minimum does not exist, it is defined by the infimum.}.  
\begin{align}\nonumber
\Psi_*&:=\min_{\e\in\cE} \Psi\Big(J_{\bm \theta}[\e]\Big),\\  \label{def:opt_design0}
\e_*&:=\mathrm{arg}\min_{\e\in\cE} \Psi\Big(J_{\bm \theta}[\e]\Big).
\end{align} 
We call this optimal design $\e_*$ as a {\it $\Psi$-optimal design}. 
The optimality function $\Psi$ is assumed to satisfy the following three properties. \\
i) Isotonicity: For $J_1\ge J_2\ge 0$, $\Psi(J_1)\ge\Psi(J_2)$ holds.\\
ii) Homogeneity: For a constant $a>0$, $\Psi(a J)=\psi(a) \Psi(J)$ with $\psi(a)$ a non-negative function. \\
iii) Convexity: $\lambda\in[0,1]$, $J_1,J_2$, 
$\Psi(\lambda J_1+(1-\lambda)J_2)\le \lambda \Psi(J_1)+(1-\lambda) \Psi(J_2)$ holds. \\
In addition to the above three conditions, we often impose the following condition.\\ 
iv) Orthogonality invariance: 
For any orthogonal matrix $O$, $\Psi(J)=\Psi(O J O^t)$. 
That is an optimality criterion depends only on the eigenvalues of the Fisher information matrix.\\ 
The first condition is also known as an operator monotone function in matrix analysis. 
The second condition is needed to incorporate additivity property of the Fisher information matrix. 
In the following discussion, we assume that the optimality function $\Psi$ always satisfies these conditions 
[i), ii), iii)] unless stated otherwise. 

We list some of the popular criteria:
\begin{enumerate}
\item $A$-optimality: $\Psi_A(J)=\mathrm{Tr}\{J^{-1}\}$. 
\item $D$-optimality: $\Psi_D(J)=\Det{J^{-1}}$. 
\item $E$-optimality: $\Psi_E(J)=\ds\lambda_{\max}(J^{-1})$ with $\lambda_{\max}$ the maximum eigenvalue.
\item $c$-optimality: $\Psi_c(J)=\bm{c}^tJ^{-1}\bm{c}$ for a given vector $\bm{c}\in\bbr^n$. 
\item $\gamma$-optimal: $\Psi_\gamma(J)=(\frac1n\mathrm{Tr}\{J^{-\gamma}\})^{1/\gamma}$ ($\gamma\in\bbr$). 
\end{enumerate}
Here $\gamma$ is a fixed parameter, and $n=|\Theta|$ is the dimension of the parameter set $\Theta$. 
It is easy to observe that $\gamma$-optimal includes the first three optimality criteria as follows. 
First, $A$-optimality is obtained as $\Psi_A(J)=n^{\gamma}\lim_{\gamma\to1}\Psi_\gamma(J)$ when we set $\gamma=1$). Second, $D$-optimality corresponds to the limit $\gamma\to0$ as 
$\Psi_D(J)=\big(\lim_{\gamma\to0} \Psi_\gamma(J)\big)^{n}$. 
Last, $E$-optimality is related to the limit $\gamma\to\infty$, 
since $\lim_{\gamma\to0} \Psi_\gamma(J)=\lambda_{\max}(J^{-1})$. 
Here, we make a brief comment on the case of non-invertible Fisher information matrix. 
It is clear that when $J$ is not full rank, then a certain regularization is needed to invert $J$. 
In this study, we mainly focus on finding optimal designs that gives raise to a regular statistical model. 
That is to guarantee matrix inversion of the Fisher information matrix. 
(We will come back to this point in Sec.~\ref{sec:misc}.)

Other than the above optimality criterion, there is a special optimal design known as the L\"owner optimality. 
This is defined by the existence of a design $\e_L$ such that 
the matrix inequality $J_{\bm \theta}[\e_L]\ge J_{\bm \theta}[\e]$ holds for all other designs $\e$. 
In general, the L\"owner optimal design does not exists.\cite{pukelsheim} 
If so, in fact, it dominates all other $\Psi$-optimality criterion due to isotonicity of a function $\Psi$. 
It is straightforward to show that the necessary and sufficient condition for the existence of 
the L\"owner optimal design; the $c$-optimal design is independent of $\bm{c}$ for all $\bm{c}\in\bbr^n$. 
\cite{pukelsheim} 

Convexity about the design space $\cE$ is of importance to find an optimal design. 
We introduce a convex sum of two designs $\e_1,\e_2\in\cE$ is defined as 
$\e_\lambda=\lambda \e_1+(1-\lambda) \e_2\in\cE$. Then, the design space $\cE$ becomes a convex set. 
With a proper definition of convex sum of two designs, it is straightforward 
to check it preserves the locally unbiasedness condition. 
In other words, if an estimator $\hat{\bm{\theta}}$ is locally unbiased at $\bm{\theta}$ under $\e_1$ and $\e_2$, then $\hat{\bm{\theta}}$ is also locally unbiased at $\bm{\theta}$ for $\e_\lambda$. 
Convexity of the optimality function states that the inequality, 
$\Psi(\lambda J_{\bm \theta}[\e_1]+(1-\lambda)J_{\bm \theta}[\e_2])\le \lambda \Psi(J_{\bm \theta}[\e_1])+(1-\lambda) \Psi(J_{\bm \theta}[\e_2])$, 
holds for $\e_1, \e_2\in\cE$ and $\lambda\in[0,1]$.  
It is easy to check that $A$-, $E$-optimality satisfy this convexity condition. 
However, $D$-optimality violates this condition and the standard remedy 
is to optimize $\log  \Det{J_{\bm \theta}[\e]^{-1}}=- \log  \Det{J_{\bm \theta}[\e]}$ in stead, which is a convex function. 
With these additional structures, we can formulate our problem 
as a convex optimization problem over the convex set. 
This problem can then be implemented efficiently in an appropriate convex optimization algorithm. \cite{fedorov,fh97,pp13,fl14} 

\subsection{Discrete design problem}
In this subsection, we extend an estimation strategy for a situation of $N$ repetition of experiments, 
there are two distinct strategies as follows. 

\noindent
\underline{I.~i.~d.~strategy:} This strategy corresponds to repeating 
exactly same design $\e$ for $N$ times, whose design is denoted by $\e^N\in\cE^N$. 
The probability distribution for this case corresponds to independently and identically distributed (i.~i.~d.) one as 
\begin{equation*}
p_{\bm \theta}(x^N|\e^N)=\prod _{t=1}^N p_{\bm \theta}(x_t|\e).
\end{equation*}
Additivity of the Fisher information matrix applies to get the Fisher information matrix $J_{\bm \theta}[\e^N]=N J_{\bm \theta}[\e]$. 
Thus, the problem is reduced to the case $N=1$. 

\noindent
\underline{Mixed strategy:} Let $N(m)$ be an $m$-partition of an integer $N$, 
i.e., $N(m)=(n_1,n_2,\dots,n_m)$ such that $\sum_{i=1}^mn_i=N$ and $n_i\ge0$. 
The mixed strategy is to repeat a design $\e_1$ for $n_1$ times, $\e_2$ for $n_2$ times, 
$\dots$, and $\e_m$ for $n_m$ times. ($N$ experiments in total.) 
The design for this mixed strategy is specified by 
the sets $(\V{\p},\V{\e})$ where $\V{\p}=(\p_1,\dots,\p_m)$ and $\V{\e}=(\e_1,\dots,\e_m)$ are 
vectors of the relative frequency and designs, respectively. 
This mixed strategy is denoted by $\e[N(m)]=(\V{\p},\V{\e})$. 
The probability distribution for the design $\e[N(m)]$ is 
\begin{equation*}
p_{\bm \theta}(x^N|\e[N(m)])=\prod _{i=1}^{m} p_{\bm \theta}(x^{n_i}|\e_{i}^{n_i})=\prod _{i=1}^{m} \prod_{t_i=1}^{n_i}p_{\bm \theta}(x_{t_i}|\e_{i}). 
\end{equation*}
The normalized Fisher information matrix, which is divided by $N$, for the design $\e[N(m)]$ is 
\begin{equation}\label{FI:exact}
J_{\bm \theta}\big[\e[N(m)]\big]=\frac{1}{N}\sum_{i=1}^m n_i J_{\bm \theta}[\e_i]. 
\end{equation}
The problem is now to find the best partition $N(m)$ for a given $N$ 
and a set of designs $\V{\e}=(\e_1,\dots,\e_m)$ such that 
the value of the $\Psi$ optimality function $\Psi(J_{\bm \theta}\big[\e[N(m)]\big])$ is minimized. 
This {\it discrete design problem}, also known as {\it exact design problem}, 
is of importance in practice to find the best design. 
However, this is a combinatoric optimization problem, 
and it is a hard problem even numerically. 
Thus, one has to find an approximated optimal solution to the problem at hand, which will be given in the next subsection. 

\subsection{Continuous design problem}\label{sec:contDoE}
When the sample size $N$ is large enough, 
we approximate the exact design problem by taking the limit $N\to\infty$ with fixed ratios $p_i=\lim_{N\to\infty} (n_i/N)$ in Eq.~\eqref{FI:exact}. 
This optimization problem is called the {\it continuous design problem} or the {\it approximated design problem}. 
In general, the optimal continuous design is a good approximation to the exact design problem 
for sufficiently large $N$. 

The problem here is to find an optimal relative frequency $\V{\p}=(\p_i)\in\cP(m)$ 
($:=$ a set of probability vector for $m$ events) and a set of designs $\V{\e}=(\e_i)\in\cE^m$ 
such that the value of a given optimality function $\Psi(J_{\bm \theta}\big[\e(m)]\big])$ is minimized. 
Here, we denote the design of this continuous design problem by 
\begin{equation}
\e(m)=(\V{\p},\V{\e})\in\cP(m)\times\cE^m. 
\end{equation} 
The Fisher information matrix about the design $\e(m)$ takes the convex mixture of each Fisher information matrix as
\begin{equation}\label{eq:CFI_random}
J_{\bm \theta}[\e(m)]=\sum_{i=1}^m \p_iJ_{\bm \theta}[\e_i]. 
\end{equation}
We can also state that this is equivalent to the Fisher information 
about the joint probability distribution: $\p_ip_{\bm \theta}(x|\e_i)$.  
In other words, the mixed strategy is to consider a statistical model, 
\begin{equation} \label{eq:jointprob}
M\left(\e(m)\right)=\Big\{\p_i p_{\bm \theta}(\cdot|\e_i)\,\Big|\,{\bm \theta}\in\Theta\Big\},
\end{equation}
with known $\p_i$, which is also to be optimized. 

Summarizing above arguments, the following optimization problem needs to be solved: 
Given an optimality function and integer $m$, to find an optimal design 
$\e_*(m)=(\V{\p}_*,\V{\e}_*)$ defined by 
\begin{equation}\label{opt_design}
\e_*(m)=\arg\hspace{-5mm}\min_{\e(m)\in\cP(m)\times\cE^m} \Psi\Big(\sum_{i=1}^m \p_iJ_{\bm \theta}[\e_i]\Big). 
\end{equation}

A convex structure for mixed strategies is naturally constructed from 
two continuous designs $\e(m)=(\V{\p},\V{\e}), \e'(m)=(\V{\p}',\V{\e}')$ as 
\begin{equation*}
\lambda \e(m)+(1-\lambda) \e'(m)=\big(\lambda\V{\p}+(1-\lambda)\V{\p}',\, \lambda\V{\e}+(1-\lambda)\V{\e}' \big),
\end{equation*}
where $\lambda\V{\e}+(1-\lambda)\V{\e}' =\big(\lambda\e_i+(1-\lambda)\e'_i \big)$ is a well-defined convex sum of two designs. 
However, it is not unique how to define a convex sum of two designs $\e(m)$ and $\e'(m')$ for $m\neq m'$ in general. 
The usual treatment of this difficulty is to introduce a measure $\xi$ on the design space $\cE$. 
This is to consider an experimental design of the form $\e_\xi:= \xi(d\e)$. 
This formalism is certainly more general, since a discrete measure reduces to the 
case of the above continuous design problem. 
The Fisher information about this design then takes of the form:
\begin{equation}
J(\xi):=\int \xi(d\e)J_{\bm \theta}[\e]. 
\end{equation}
With this notation, the problem is expressed as 
\begin{align} \nonumber
\Psi_*&=\min_{\xi\in \Xi} \Psi\Big(J(\xi)\Big), \\
\label{opt_design2}
\xi_*&=\arg\min_{\xi\in \Xi} \Psi\Big(J(\xi)\Big), 
\end{align}
with $\Xi$ the totality of probability measures on the design space $\cE$. 
We call $\xi$ as the {\it design measure} or simply a {\it design} when no confusion arises.  
This is an object of our interest in the theory of optimal DoE. 
It is easy to see from Carath\'eodory's theorem that 
an optimal design measure can be found by using not more than $ n(n+1)/2+1$ support points 
with $n$ the number of parameters to be estimated. 

Another important problem other than finding an optimal design is to characterize 
the structure of the Fisher information matrix for all possible designs: 
\begin{align}\nonumber
\cJ(\cE)&:= \{J_{\bm \theta}[\e]\,|\, \e\in\cE \}, \\ \label{Fisherset}
\cJ(\Xi)&:= \{\int \xi(d\e)J_{\bm \theta}[\e]\,|\, \xi\in\Xi \}.
\end{align}
Clearly, $\cJ(\Xi)$ is the convex hull of $\cJ(\cE)$. 
We call the sets $\cJ(\cE)$ and $\cJ(\Xi)$ as the {\it Fisher information regions}.
This Fisher information region is a well-known concept in classical statistics. 
Several authors  studied it in the context of the quantum estimation theory.\cite{hayashi2016quantum,yamagata11,zhu2015information} 
The optimization problem takes the following alternative form as a minimization over the convex set: 
\begin{align}\nonumber
\Psi_*&:=\min_{J\in\cJ} \Psi(J),\\ 
J_*&:=\arg\min_{J\in\cJ} \Psi(J). 
\end{align} 
From the optimal Fisher information matrix, we then associate it with the optimal design as $J_{\bm \theta}[\e_*]=J_*$. 

\subsection{Necessary and sufficient condition}\label{sec-optcondition}
Under the assumptions made in our discussion, we can derive the necessary and sufficient condition 
for the optimal design in various different forms. 
This is one of the central subject in the theory of optimal DoE.\cite{fedorov,pukelsheim,fh97,pp13,fl14} 

For a given optimality function $\Psi$ satisfying conditions in Sec.~\ref{sec:formulation}, 
we consider the directional derivative: 
\[
\Psi'(\xi_0;\xi):=\lim_{\epsilon\to0}\frac{1}{\epsilon}\Big[\Psi\big((1-\epsilon)J(\xi_0)+\epsilon J(\xi) \big)-J(\xi_0)\Big],
\]
where $J(\xi)=\int \xi(d\e)J_{\bm \theta}[\e]$. 
It is straightforward to see that $\xi_*$ is an optimal design measure if and only if 
the directional derivative is nonnegative $\Psi'(\xi_*;\xi)\ge 0$ for all $\xi\in\Xi$. 
In the theory of optimal DoE, many of the optimality functions admit the following 
special form for the directional derivative: 
\[
\Psi'(\xi_0;\xi)=\int \xi(d\e) \psi(\e,\xi_0). 
\]
with $\psi(\e,\xi)$ some function of the design measure and design. 
It is convenient to introduce the {\it sensitivity function} $\varphi$ for the optimality function $\Psi$ by 
\be
\varphi(\e,\xi):=-\psi(\e,\xi)+C(\xi),
\ee
where $C$ is a function of the design measure $\xi$, defined for each $\Psi$.\cite{fh97,fl14} 
As examples, let us list two popular optimal designs; the $A$-optimality $\Psi(J)=\mathrm{Tr}\{WJ^{-1}\}$ with a weight matrix $W>0$ 
and $D$-optimality with $\Psi(J)=\log \Det{J^{-1}}$. \\
\underline{$A$-optimality:} \\
$\psi(\e,\xi)=\Tr{W J(\xi)^{-1} J_{\bm \theta}[\e]J(\xi)^{-1}}$, $C(\xi)=\Psi(\xi)$, \\
where $\Psi(\xi)= \Psi(J(\xi))$.\\
\underline{$D$-optimality:}\\
$\psi(\e,\xi)=\Tr{J(\xi)J_{\bm \theta}[\e]}$, $C(\xi)=n$. 

With these notations, the following theorem holds.\cite{fedorov,pukelsheim,fh97,fl14} 
\begin{theorem} \label{thm_equiv2}
For an optimality function satisfying the condition discussed before, 
the following design problems are equivalent. 
\begin{align*}
\mathrm{1)}&\quad \min_{\xi\in\Xi}\Psi(\xi),\hspace{5cm}\\
\mathrm{2)}&\quad \min_{\xi\in\Xi}\max_{\e\in\cE}\varphi(\e,\xi),\\
\mathrm{3)}&\quad \max_{\e\in\cE}\varphi(\e,\xi)=\Psi(\xi). 
\end{align*}
\end{theorem}
This theorem can be regarded as a generalization of one of the celebrated results in the theory of optimal DoE, 
known as the equivalence theorem due to Kiefer and Wolfowitz.\cite{kw60} See Refs.~\cite{fedorov,pukelsheim,fh97,pp13,fl14} for more details. 

\subsection{Miscellaneous items}\label{sec:misc}
To supply some more languages for optimal DoE, we list a few of them.  
First, we need to be clear on the concept of optimality in the optimal DoE. 
Let us consider an optimization problem \eqref{def:opt_design0} for simplicity. 
More general case of optimal designs can be considered similarly. 

The optimal design $\e_*=\mathrm{arg}\min_{\e\in\cE} \Psi(J_{\bm \theta}[\e])$ is called 
a {\it local optimal design} in the sense that it is optimal at a specific point ${\bm \theta}$. 
In general, this local optimal design depends on the unknown value ${\bm \theta}$. 
In other words, we should express it as $\e_{*}(\bm{\theta})=\mathrm{arg}\min_{\e\in\cE} \Psi\Big(J_{\bm \theta}[\e]\Big)$. 
When dealing with the generic statistical models, one always finds a local optimal design only. 
Only when, one simplifies a model, such a simple linear regression model, we can find the global optimal design, which is optimal uniformly in $\bm{\theta}$, i.e., $\forall \bm{\theta},\bm{\theta}',\e_{*}(\bm{\theta})=\e_{*}(\bm{\theta}')$.  
In practice, one then has to combine other techniques of DoE to realize an optimal design. 
This has been studied in the field of classical optimal DoE in past 
under the name of the adaptive or the sequential design problem. \cite{fedorov,pukelsheim,fh97,pp13,fl14}
The adaptive estimation scheme will not be a subject of our paper due to the page limitation. 
It is interesting to lean that these adaptive schemes were independently discovered in the context of quantum state estimation problems. Nagaoka first proposed such an adaptive method based on updating the likelihood function.\cite{nagaoka89-2} 
Later, others proposed different variants of adaptive methods.\cite{HM98,BNG00} 
The latter method is based on splitting $N$ samples into two sets. 
The first set is used to give a rough estimate, and then we apply a near optimal strategy for the second set. 
We note that this method was already well studied in the classical statistics.\cite{fedorov,pukelsheim,fh97,pp13,fl14} 
As a word of caution, the two-step method for the asymptotic case is a method of proof for convergence. 
A practical problem in the theory of DoE is to find the optimal division of $N$ samples into two sets or more generally several sets, which gives the lowest estimation error. 

Second, we say that a design $\e$ is {\it singular}, when the resulting statistical model 
$M(\e)$ is not regular. 
See for example Ref.~\cite{at95} on the detail discussion of non-regular models. 
One common instance of a singular design is when the classical Fisher information matrix is rank deficient. In fact, we often deal with singular models in the theory of optimal DoE. 
In this case, we may use the generalized inverse matrix method to evaluate the inverse of the classical Fisher information matrix. However, we cannot estimate all parameters simultaneously. 
There are alternative techniques known in the theory of optimal DoE.\cite{fedorov,fh97,pp13,fl14} 
Appendix Sec.~5 of Ref.~\cite{gns19} gives a short summary for these methodologies. 
We will make a few more comments on local optimality and the problem of singular designs in Sec.~\ref{sec:local_singular} for the quantum case. 

In passing, we note that a recent paper \cite{sp2020_ijqi} discussed non-regular measurements. 
They called a measurement $\Pi$ (a design $\e$ in our terminology) is regular, when it is $\bm{\theta}$-independent. 
We stress that $\bm{\theta}$-independence is different from the concept of local optimal design. 
Further, they introduced a non-regular measurement $\e_{\bm \theta}$ that also comprises a part of parametric dependence in the resulting statistical model: 
\begin{equation}
M(\e_{\bm \theta})=\{ p_{\bm \theta}(\cdot|\e_{\bm \theta})\,|\,\bm{\theta}\in\Theta\}. 
\end{equation}
We note that this setting is unusual in the sense that the design $\e$ is no longer under our control. 
It is rather a part of the statistical model itself under consideration. 
In this special case, one has to differentiate $\bm{\theta}$ for the family of designs $\e_{\bm \theta}$, since we do not have precise knowledge on it. 

Third, a family of states $\{s_{\bm \theta}\,|\, {\bm \theta}\in\Theta\}$ is said locally {\it identifiable} at $\bm{\theta}_0$, if there exists some neighborhood $\cB_{\bm{\theta}_0}$ of $\bm{\theta}_0$ such that 
the following conditions is satisfied:
\begin{equation}
\forall\bm{\theta}\in\cB_{\bm{\theta}_0},\forall\e\in\cE,\,p_{\bm \theta}(\cdot|\e)=p_{\bm{\theta}_0}(\cdot|\e)\Rightarrow \bm{\theta}=\bm{\theta}_0.
\end{equation} 
When this property holds for all parameter set, i.e., $\cB_{\bm{\theta}_0}=\Theta$, we say this family is (globally) {\it identifiable}. 
Clearly, if statistical models $M(\e)$ for all designs $\e\in\cE$ are regular, 
$\bm{\theta}\mapsto p_{\bm \theta}(\cdot|\e)$ is one-to-one. 
Thus, the identifiability condition is satisfied. 

In addition to identifiability of states, we have an issue of estimability. 
It is easy to check that we cannot estimate all parameters when a design is singular. 
In this case, only a certain linear combinations of the parameters can be estimated by this singular design. 
In the following, we focus on the case of a linear combination of the parameters. 
See for example Refs.~\cite{pukelsheim,pp13} for more general case.
Suppose one is only interested in estimating a linear combination of parameters: 
\begin{equation}\label{def:theta-c}
\theta_{\bm{c}}:=\bm{c}^t\bm{\theta}=\sum_{i=1}^nc_i\theta_i, 
\end{equation}
for a given $n$-dimensional (column) vector $\bm{c}$. 
In the language of optimal DoE, this setting is the $c$-optimal design problem. 
The parameter ${\theta}_{\bm{c}}$ is said {\it estimable}, 
if there exists a design $\e$ such that the range of $J_{\bm{\theta}}[\e]$ includes the vector $\bm{c}$. Otherwise, the design $\e$ cannot be use to estimate ${\theta}_{\bm{c}}$. 
We can also express this condition by the concept of the feasibility cone as follows. 
Define the {\it feasibility cone} for $\bm{c}$ by the subset of non-negative matrices:  
\begin{equation}\label{def:feasibility_cone}
\cA(\bm{c}):=\{ A\in\bbr^{n\times n}\,|\, A\ge0,\,A\bm{c}\neq 0\}. 
\end{equation}
Then, ${\theta}_{\bm{c}}$ is estimable if and only if $J_{\bm{\theta}}[\e] \in \cA(\bm{c})$ for some design $\e$. 
Therefore, the $c$-optimal design problem should be reformulated as
\begin{equation}\label{def:c-opt_singular}
\e_*:=\mathrm{arg}\hspace{-12pt}\min_{\e\in\cE:J_{\bm{\theta}}[\e] \in \cA(\bm{c})} c^t\big(J_{\bm{\theta}}[\e]\big)^{-1}c. 
\end{equation} 
Here, the inverse of the Fisher information matrix is evaluated in the sense of the generalized inverse. 

As a final remark on the singular design problem, we make a comment on the optimal DoE. 
The $E$-optimal design problem is also expressed as the following alternative form:
\begin{align*}
\Psi_E(J)&=\ds\lambda_{\max}(J^{-1})=\max_{\bm{c}\in\bbr^n:|\bm{c}|=1}\bm{c}^tJ^{-1}\bm{c}\\
&=\ds\left(\lambda_{\min}(J)\right)^{-1}
=\left(\min_{\bm{c}\in\bbr^n:|\bm{c}|=1}\bm{c}^tJ\bm{c}\right)^{-1}. 
\end{align*}
From this expression, we see that $E$-optimal design 
is amount to the min-max optimization of a certain $c$-optimal design problem. 
Unlike to the standard $c$-optimality criterion, 
however, we are interested in estimating all parameters in the $E$-optimality criterion. 
Therefore, we should avoid singular optimal designs. 

Related to the issue of local optimal designs and singular designs, we have a remark on the value of the optimality function. 
Let us denote the minimum value of an optimality function at $\bm{\theta}$ by $\Psi_{\bm \theta}$. 
Consider arbitrary two-different points $\bm{\theta}_0$ and $\bm{\theta}_0'$, 
and corresponding optimal designs:
\begin{align*}
\e_{*}(\bm{\theta}_0)&:=\mathrm{arg}\min_{\e\in\cE} \Psi\Big(J_{\bm \theta}[\e]\Big)\\
\e_{*}(\bm{\theta}_0')&:=\mathrm{arg}\min_{\e\in\cE} \Psi\Big(J_{\bm{\theta}_0'}[\e]\Big). 
\end{align*}
The design $\e_{*}(\bm{\theta}_0)$ is optimal at $\bm{\theta}_0$, but not at $\bm{\theta}_0'$. 
Generally speaking, there is no ordering relation between two values $\Psi_{\bm{\theta}_0}$ and $\Psi_{\bm{\theta}_0'}$, nor matrix ordering between two optimal Fisher information matrices, 
$J_{\bm{\theta}_0}[\e_{*}(\bm{\theta}_0)]$ and $J_{\bm{\theta}_0'}[\e_{*}(\bm{\theta}_0')]$. 
To see this, let us consider the case when two points are nearby each other. 
In this case, a small deviation $\bm{\theta}_0'=\bm{\theta}_0+\bm{\delta}$, with $\bm{\delta}=(\delta_i)$ a small vector, 
results in the following approximation up to the first order in $|\bm{\delta}|$:
\begin{align}
J_{\bm{\theta}_0'}[\e_{*}(\bm{\theta}_0')]&= J_{\bm{\theta}_0+\bm{\delta}}[\e_{*}(\bm{\theta}_0')]\\ 
\label{eq:approx_opt}&\simeq J_{\bm{\theta}_0}[\e_{*}(\bm{\theta}_0')]+\left.\sum_k\delta_k\frac{\del J_{\bm{\theta}}[\e_{*}(\bm{\theta}_0')]}{\del\theta_k}\right|_{\bm{\theta}=\bm{\theta}_0}, 
\end{align}
where the partial differentiation of a matrix is done by component wise. 
The second term of Eq.~\eqref{eq:approx_opt} is symmetric, but it does not have a definite sign as a matrix in general. 
Upon assuming closeness between two designs, we substitute a relation, $\e_{*}(\bm{\theta}_0')\simeq(1-\epsilon)\e_{*}(\bm{\theta}_0)+\epsilon\e_0$ 
for some design $\e_0$ and small $\epsilon$ in the sense of a randomized design. 
Then, the first term of Eq.~\eqref{eq:approx_opt} is expressed as
\begin{align}
J_{\bm{\theta}_0}[\e_{*}(\bm{\theta}_0')]&\simeq J_{\bm{\theta}_0}[(1-\epsilon)\e_{*}(\bm{\theta}_0)+\epsilon\e_0]\\
&=(1-\epsilon)J_{\bm{\theta}_0}[\e_{*}(\bm{\theta}_0)]+\epsilon J_{\bm{\theta}_0}[\e_0]. 
\end{align}
Therefore, we obtain an approximated relationship between two Fisher information matrices 
$J_{\bm{\theta}_0}[\e_{*}(\bm{\theta}_0)]$ and $J_{\bm{\theta}_0'}[\e_{*}(\bm{\theta}_0')]$ 
without a definite matrix ordering.  

As a final remark on the singular design problem, we comment on the use of generalized inverse of the Fisher information matrix. 
When an optimal design $\e_{*}$ is singular, an extra complication may arise. 
In this case, we often use an appropriate generalized inverse of the Fisher information matrix for $J_{\bm{\theta}}[\e_{*}]$. 
In some circumstances, the obtained result may depend on a particular choice of generalized inverses. See for example Refs.~\cite{pukelsheim,pp13}. 
This point will be important for the $c$-optimality for example. 
We will expand this discussion in Sec.~\ref{sec:local_singular} for the quantum case. 

To end this subsection, we shortly list three extensions of the theory of optimal DoE. 
First is the optimal design under constraint(s). 
The design is typically subject to an additional constraint(s) 
in order to take into account realistic experimental situations. Optimal design 
of experiments with constraint(s) can also be formulated.\cite{fh97,fl14} 

Second is a compound optimal design. 
Consider two optimality functions to define a new function $\Psi_\nu:=\nu \Psi_1+(1-\nu) \Psi_2$ with $\nu$ fixed positive parameters. 
$\e_*=\arg\min\Psi_\nu[\e]$ is called a {\it compound optimal design}, 
and it represents a tradeoff relation between two different optimal designs defined by $\Psi_1,\Psi_2$. 

Last is to evaluate efficiency of design. 
Given an optimality function $\Psi$, we can define the optimal design $\e_*$ for this optimality. 
In practice, one is not only interested in finding the optimal design, 
but also the performance of a suboptimal design, say $\e_{0}$, which can be easily implemented. 
To this end, we need to know smallness of the value $\Psi(\e_{0})$. 
Note that $\Psi[\e]$ is a relative quantity, and hence, we cannot immediately conclude the performance of the design $\e_{0}$ based on the value $\Psi(\e_{0})$ only. 
The standard way to handle this problem is to consider a normalized version function $\Psi$ by its optimal value $\Psi[\e_*]$, which is defined as 
\begin{equation}
\eta_\Psi[\e]:=\frac{\Psi[\e_*]}{\Psi[\e]}=\frac{\min_{\e\in\cE}\Psi[\e]}{\Psi[\e]} . 
\end{equation}
We call $\eta_\Psi[\e]$ the {\it efficiency} of the design $\e$ with respect to the optimality function $\Psi$. By definition, the normalized function satisfies $0\le\eta_\Psi[\e]\le1$. 
Notably, the equality $\Psi[\e]=1$ does not necessary imply $\e$ is an optimal design for $\Psi$ optimality. 

An another application of efficiency of design is comparison of different optimality criteria. 
Consider two optimality criteria based on $\Psi_1$ and $\Psi_2$. 
The optimal design $\e_{*}=\arg\min\Psi_1[\e]$ is optimal for $\Psi_1$ but not for $\Psi_2$. 
One may naively expect that this is also good for $\Psi_2$. 
To quantify how good it is, we can analyze efficiency 
\begin{equation}
\eta_{\Psi_2}[\e_*]=\frac{\min_{\e}\Psi_2[\e]}{\Psi_2[\e_*]} . 
\end{equation}
If this quantity is close to $1$, it means that $\e_*$ is also good for the other criterion.  
This will be studied in Sec.~\ref{sec5}. 

Applications of the above extended optimal designs were discussed 
in various statistical problems, see Refs.~\cite{pukelsheim,fh97,pp13,fl14}. 

\section{Quantum state estimation as optimal design of experiments}\label{sec3}
We now apply the theory of optimal DoE to the parameter estimation problem in quantum systems. 

\subsection{Definitions}
A {\it quantum system} is represented by a $d$-dimensional complex vector space $\bbc^d$. 
With the standard inner product, it becomes a Hilbert space denoted by $\cH=\bbc^d$. 
When the dimension of the system is two, we speak of ``qubit" that is the simplest quantum system. 
To simplify our discussion we only consider quantum systems with a fixed dimension $d<\infty$. 
A {\it quantum state} $\rho$ is a non-negative matrix on $\cH$ with unit trace. 
The set of all quantum states on $\cH$ is denoted by 
$\sofh:=\{\rho\in\bbc^{d\times d}\,|\,\rho\ge0,\tr{\rho}=1 \}$.
A {\it measurement} $\Pi$ on a given quantum state $\rho$ is described a set of positive semidefinite matrices 
$\Pi=\{\Pi_x\}_{x\in\cX}$ ($\forall x,\Pi_x\ge0$) such that the condition $\sum_{x\in\cX}\Pi_x=I_d$ (Identity matrix) is satisfied. 
When $\Pi$ is performed on $\rho$, the measurement outcomes are drawn 
according to a probability distribution: 
\begin{equation*}
p_\rho(x|\Pi):=\tr{\rho\Pi_x}. 
\end{equation*}
Here the set $\cX$ is a label set for the measurement outcomes. 
This probabilistic rule (Born's rule) will be used to define the model function. 

\subsection{Formulation of the problem}
We are now in place to formulate the parameter estimation problem about quantum states as a problem of an optimal DoE . 
Given a family of $n$-parameter quantum states  
\begin{equation*}
M^Q:=\{\rho_{\bm \theta}\,|\,{\bm \theta}\in\Theta\subset\bbr^n\}, 
\end{equation*}
under the assumption that ${\bm \theta}\mapsto\rho_{\bm \theta}$ is one-to-one and smooth mapping. 
We identify the quantum state $\rho$ as the state $s$. 
The design in our setting is a measurement $\e=\Pi$, and the model function is given by Born's rule: 
\begin{equation*}
f:\,(\rho_{\bm \theta},\e)\mapsto p_{\bm \theta}(x|\e)=\tr{\rho_{\bm \theta} \Pi_x}\quad({x\in\cX}). 
\end{equation*}
Thus, the design space is the set of all possible POVMs. 

The statistical model for a design $\e$ is obtained as 
\begin{equation*}
M(\e)=\{p_{\bm \theta}(\cdot|\e)\,|\,\bm{\theta}\in\Theta\}. 
\end{equation*}
We wish to find an optimal design $\xi_*\in\Xi$ 
that minimizes a properly chosen optimality criterion as Eq.~\eqref{opt_design2}. 
An important aspect of the optimal design problem for quantum state estimation  
is that the design $\e=\Pi$ (measurement) is subject to the constraints: 
\[
\forall x,\,\Pi_x\ge0\mbox{ and }\sum_{x\in\cX}\Pi_x=I_d, 
\]
that gives rise to $d(d+1)/2$ constraints for positive semidefinite matrices $\Pi_x$.  
A unique feature of DoE in the quantum case is 
that these constraints appear in the design space $\cE$ by the laws of quantum theory. 

As stated before, convex structure in the design space (a set of all POVMs) is important. 
A convex mixture of two POVMs is defined as follows. 
Let $\Pi=\{\Pi_1,\Pi_2,\ldots,\Pi_k\}$ and $\Pi'=\{\Pi'_1,\Pi'_2,\ldots,\Pi'_{k'}\}$ be two POVMs. 
For a given $\lambda\in[0,1]$, we define a new POVM by
\begin{align*}
\Pi_\lambda&=\lambda\Pi\cup(1-\lambda)\Pi'\\
&:=\{\lambda\Pi_1,\lambda\Pi_2,\ldots,\lambda\Pi_k\}\cup\{(1-\lambda)\Pi'_1,(1-\lambda)\Pi'_2,\ldots,(1-\lambda)\Pi'_{k'}\}\\
&=\{\lambda\Pi_1,\lambda\Pi_2,\ldots,\lambda\Pi_k,(1-\lambda)\Pi'_1,(1-\lambda)\Pi'_2,\ldots,(1-\lambda)\Pi'_{k'}\}, 
\end{align*}
whose measurement outcomes are $k+k'$. 
The convex structure for the POVM space plays an important role, 
since the problem can be casted into a convex optimization problem. 
This point was already pointed out in the literature.\cite{D_Ariano_2005,fujiwara06,yamagata11}

When some of POVM elements are proportional to each other, 
one could combine them without affecting measurement statistics. 
For example, assume $\Pi'_1=c\Pi_1$ with $c$ a positive constant, then a new POVM element of the form $\Pi^{\rm new}_1=\big(\lambda+c(1-\lambda)\big)\Pi_1$ provides the same design. 

\subsection{Extensions of the problem}
In this subsection, we briefly list possible extensions of the DoE formalism for the quantum-state estimation problem. We note that most of these results are already known in the literature, yet we could present them in a unified manner based on the language of the theory of optimal DoE.
 
\subsubsection{Restricted measurement}\label{sec:restricted_POVM}
When only some of measurements are accessible in laboratory, it does not make sense to find an optimal POVM among all possible POVMs. Let $\cE_0\subset \cE$ be the subset of the design space, and consider the following optimization problem\footnote{Even if the minimum exists for the original design problem, the restricted design problem may only allow the solution in the sense of the infimum.}: 
\begin{align}\nonumber
\Psi^0_*&:=\min_{\e\in\cE_0} \Psi\Big(J_{\bm \theta}[\e]\Big),\\ 
\e^0_*&:=\mathrm{arg}\min_{\e\in\cE_0} \Psi\Big(J_{\bm \theta}[\e]\Big).
\end{align} 
Clearly, this optimal design $\e_*^0$ represents what we could best among all ``accessible" POVMs. 
A typical case is when only projections measurements are allowed. 
In this case, we optimize over the PVM space $\cE_{\rm PVM}=\{ \Pi|\mbox{$\Pi$ is PVM.}\}$. 
We know that an optimal measurement is not in general given by a PVM. 
By considering the continuous design problem, we could do better in general. 
The problem to be solved now is 
\begin{align}\nonumber
\Psi_*^{\rm Random}(m)&:=\min_{\e(m)\in\in\cP(m)\times\cE_{\rm PVM}^m} \Psi\Big(J_{\bm \theta}[\e]\Big),\\ 
\e_*^{\rm Random}(m)&:=\mathrm{arg}\min_{\e(m)\in\in\cP(m)\times\cE_{\rm PVM}^m} \Psi\Big(J_{\bm \theta}[\e]\Big).
\end{align}
Then, we wish find an optimal $m_*$ minimizing the number of different designs. 
By randomizing different designs, the optimal design $\e_*(m_*)$ can perform better; $\Psi_*^{\rm Random}(m_*)\le \Psi^0_*$. 

Note that one could attain optimal precision in some case by solving the above continuous design problem within the restricted design space. 
In other words, one could do best simply by measuring several PVMs randomly according to a proper distribution. In Sec.~\ref{sec:qubit-FIregion}, we will show that all possible qubit models can exhibit such an optimal solution. In higher dimensional case, it seems that this is not the case.  
In Sec.~\ref{sec:FI}, we give more discussion on this point.

\subsubsection{Classical-quantum state formalism}\label{sec:cq-state}
The continuous design problem is also interpreted as follows. 
The basic idea is to use a classical-quantum (CQ) state. 
Let us rewrite a quantum state $\rho_{\bm \theta}$ as a random mixture of an extended state of the form:
\begin{equation}
\widehat{\rho}_{\bm \theta}=\sum_{i=1}^m \p_i \cket{i}\cbra{i}\otimes \rho_{\bm \theta}, 
\end{equation}
where $\{\cket{i}\}_{i=1}^{m}$ is an orthonormal basis for the $m$-dimensional real vector space $\cH_C:=\bbr^m$, 
and $\p=(\p_i)$ denotes the known probability vector. 
Thus, the total Hilbert space is extended to $\widehat{\cH}=\cH_C\otimes\cH$. 
Next, we consider a set of POVMs $\e(m)=(\Pi^{(1)},\Pi^{(2)},\ldots,\Pi^{(m)})$, 
whose element forms a valid POVM $\Pi^{(k)}=\{\Pi^{(k)}_x\}_{x\in\cX_k}$ for each $k$. 
If we perform a POVM on the extended space $\widehat{\cH}$ of the form $\widehat{\Pi}:=\{\widehat{\Pi}_{k,x_k} \}_{(k,x_k)}$ with $\widehat{\Pi}_{k,x_k}:=\cket{k}\cbra{k}\otimes \Pi^{(k)}_{x_k}$, 
the resulting statistical model is given by 
\begin{equation}
\widehat{M}(\widehat{\Pi})=\{\widehat{p}_{\bm \theta}(\cdot|\widehat{\Pi})\,|\,  {\bm \theta}\in\Theta\},
\end{equation}
where measurement outcomes is labeled by the double index as $\widehat{p}_{\bm \theta}(k,x_k|\widehat{\Pi})$. By construction, we have 
\begin{equation}
\widehat{p}_{\bm \theta}(k,x_k|\widehat{\Pi})=\p_k\tr{\rho_{\bm \theta}\Pi^{(k)}_{x_k}}, 
\end{equation}
which forms a joint probability distribution. [See Eq.~\eqref{eq:jointprob}.]

Additivity of the classical Fisher information matrix yields the formula:
\begin{equation}
J_{\bm \theta}[\widehat{\Pi}]= \sum_{k=1}^m \p_kJ_{\bm \theta}[\Pi^{(k)}]. 
\end{equation}
This is exactly the same formula as Eq.~\eqref{eq:CFI_random}, which is obtained as the continuous design problem. 
Although this mathematical equivalence is almost trivial, this result might come out as a surprise when interpreted as follows. 
Consider process tomography or a channel estimation problem instead in the framework of DoE.\cite{gns19} 
A task here is to design a set of good input states and send them to an unknown channel. 
Output states are then measured with appropriate POVMs. 
It is clear that we need to prepare multiple input states to find an optimal strategy. 
If we phrase the whole process as the CQ state scenario, we might then 
interpret it as if we only need to prepare a one big CQ state. 
However, we should not call it as a ``one-shot" estimation strategy. 
A trick here is, of course, we are working on the infinite sample size limit to approximate the exact design problem. 

\subsubsection{Collective measurement strategy}
It is well known that collective measurements on multiple copies of a state can perform equally or better than individual measurements depending on the nature of models. 
The case of collective strategy can also be handled similarly. 
Consider $N$ identical copies of unknown states: 
$\rho_{\bm \theta}^{\otimes N}:=\rho_{\bm \theta}\otimes \rho_{\bm \theta}\otimes\ldots\otimes \rho_{\bm \theta}$. The design now is described by a POVM on the $N$ tensor Hilbert space 
$\cH^{\otimes N}=\cH\otimes \cH\otimes\ldots\otimes\cH$. 
Then, the optimization problem is given by
\begin{align}\nonumber
\Psi^{(N)}_*&:=\min_{\e\in\cE^{(N)}} \Psi\Big(J_{\bm \theta}[\e]\Big),\\ 
\e^{(N)}_*&:=\mathrm{arg}\min_{\e\in\cE^{(N)}} \Psi\Big(J_{\bm \theta}[\e]\Big),
\end{align} 
where $\cE^{(N)}$ denotes the set of all possible POVMS on $\cH^{\otimes N}$. 

\subsubsection{Holevo-Nagaoka type bound}
In the theory of quantum state estimation, the Holevo bound \cite{holevo} established 
the fundamental precision limit. This bound is defined by minimizing a function 
of an $n\times n$ positive semi-definite matrix $Z[\vec{X}]$ whose components are
\begin{equation}
Z[\vec{X}]=\big[\tr{\rho_{\bm \theta}X_kX_j} \big]_{jk}, 
\end{equation} 
over an $n$ Hermitian operators $\vec{X}=(X_1,X_2,\ldots,X_n)$ under the locally unbiased condition: 
\begin{equation}
\cX:=\left\{\vec{X}\,|\,X_i\ \text{Hermitian},\, \forall ij,\tr{\rho_{\bm \theta} X_i}=0,\, \tr{\frac{\del \rho_{\bm \theta}}{\del \theta_i}X_j}=\delta_{ij}\right\}. 
\end{equation}
It is important to note that when $\vec{X}$ in the set $\cX$, 
the conditions $\tr{\frac{\del \rho_{\bm \theta}}{\del \theta_i}X_j}=\delta_{ij}$ 
require $X_1,X_2,\ldots,X_n$ to be linearly independent. 
And hence, $Z[\vec{X}]>0$ for all $\vec{X}\in\cX$. 
The Holevo bound sets the lowest achievable convergence rate 
in the asymptotic limit ($N\to\infty$). \cite{HM08,KG09,YFG13,YCH18}

In the language of the theory of DoE, the Holevo bound gives the first order asymptotics 
for the $A$-optimality under the collective POVM strategy explained in the previous subsection. 
It is then natural to extend the Holevo bound for other optimality criteria. 
This is done by a straightforward manner and we only provide the final result without details. 
Derivation here follows exactly same manner as Nagaoka's formulation. \cite{nagaoka89} 
We shall call the bound as the Holevo-Nagaoka type bound. 
For a given optimality function $\Psi$, the Holevo-Nagaoka type bound is given as follows. 
\begin{theorem}
The minimum value of the optimality function is bounded by
\begin{equation}
\min_{\xi\in\Xi}\Psi\big(\xi\big)\ge\Psi^{HN}.
\end{equation}
The $\Psi$-optimal Holevo-Nagaoka bound $\Psi^{HN}$ is defined by the minimization: 
\begin{align}
\Psi^{HN}&:=\min_{\vec{X}\in\cX}\chi_\Psi\big(Z[\vec{X}] \big),
\end{align}
where $\chi_\Psi$ is defined indirectly by the minimization: 
\begin{equation}
\chi_\Psi(Z):=\min_{J}\{\Psi\big( J \big)\,|\,J>0\ {\rm real},\,J^{-1}\ge Z \}.  
\end{equation} 
\end{theorem}
As an example, the $A$-optimal Holevo-Nagaoka type bound is 
\[
\chi_{\Psi_A}(Z):=\Tr{\rm{Re}\, Z}+\Tr{|\rm{Im}\, Z|}.  
\]


Another straightforward extension is to bound the Holevo-Nagaoka type bound further by quantum Fisher information matrix $J_{\rm \theta}^Q$ such as the SLD and right logarithmic derivative (RLD) Fisher information matrices. Then, we obtain 
\[
\min_{\xi\in\Xi}\Psi\big(\xi\big)\ge\Psi\big(J_{\bm \theta}^Q \big). 
\]
This also follows from the fact that the quantum Fisher information matrix dominates the classical Fisher information matrix $J_{\bm \theta}^Q\ge J_{\bm \theta}[\e]$ for all designs. 

\subsection{Fisher information region}\label{sec:FI}
As emphasized in Sec.~\ref{sec:contDoE}, the Fisher information region is a key concept upon analyzing the problem of finding optimal DoE. 
Generally speaking, it is a hard task to obtain an exact structure for the Fisher information region analytically. 
In some case, this problem is even harder to find an optimal design itself. 
Nevertheless, it is worth deriving an approximated Fisher information region $\cJ_{approx}$ 
such that the true Fisher information region is the subset. 
Such the larger set $\cJ_{approx}$ can be used to derive the lower bound for the estimation errors for the optimality function under consideration. The celebrated Gill-Massar bound \cite{gm00} was derived 
by this logic, although the concept of the Fisher information region was not utilized explicitly. 

We now discuss an important property of the Fisher information region about the quantum-state estimation problem. 
Let us define two Fisher information regions as in Eq.~\eqref{Fisherset}.  
\begin{align}\nonumber
\cJ(\cE_{\rm POVM})&:= \{J_{\bm \theta}[\e]\,|\, \e\in\cE_{\rm POVM} \}, \\ \label{QFisherSet}
\cJ(\Xi)&:= \{\int \xi(d\e)J_{\bm \theta}[\e]\,|\, \xi\in\Xi \},
\end{align}
where $\cE_{\rm POVM}$ denotes the set of all POVMs. 
By definition, $\cJ(\Xi)$ is the convex hull of $\cJ(\cE_{\rm POVM})$, and hence, $\cJ(\cE_{\rm POVM})\subset\cJ(\Xi)$ holds. 
The difference between two sets represents how much we could gain by considering randomized POVMs, 
or considering the continuous design problem in the asymptotic limit. 
It is worth emphasizing that the quantum-state estimation problem is a special case 
in the sense that there is no gap between two strategies. 
The reason behind it is that the general POVM itself contains this kind of randomized POVMs by nature. 
To summarize this result, we have the following result. 
\begin{proposition}\label{prop1}
Two Fisher information regions are identical for the quantum-state estimation problem: 
$\cJ(\cE_{\rm POVM})=\cJ(\Xi)$ 
\end{proposition}
\begin{proof}
To prove the statement, it is enough to show the inclusion relation $\cJ(\Xi)\subset\cJ(\cE_{\rm POVM})$, since the converse relation holds by definition. 
Let us consider arbitrary continuous design $\e(m)=(\V{\p},\V{\Pi})$, 
where $\V{\Pi}=(\Pi^{(1)},\Pi^{(2)},\ldots,\Pi^{(m)})$ is a set of $m$ different POVMs. 
The Fisher information matrix is expressed as 
\begin{equation}\label{eq:CFI-contDoE}
J_{\bm \theta}[\e(m)]=\sum_{k=1}^m \p_k J_{\bm \theta}[\Pi^{(k)}]. 
\end{equation} 
Next, consider the following single POVM, 
\begin{equation}
\Pi= \bigcup_{k=1}^m \p_k\Pi^{(k)}.  
\end{equation}
By construction, $\Pi$ is a convex mixture of $m$ different POVMs, which are made up of with $\sum_{k=1}^mx_k$ outcomes in total. 
It is straightforward to show that the above single POVM gives the same classical Fisher information matrix as Eq.~\eqref{eq:CFI-contDoE}. 
The case of an integral form, $\int \xi(d\e)J_{\bm \theta}[\e]$ can be done similarly by taking an appropriate limit. 
In summary, every $J\in \cJ(\Xi)$ is also in the set $\cJ(\cE_{\rm POVM})$, and thus $\cJ(\Xi)\subset\cJ(\cE_{\rm POVM})$ holds. 
\end{proof}

\subsection{Analytically solvable cases}
\subsubsection{Single (scalar) parameter model}\label{sec:1para_opt}
When the number of parameters characterizing quantum states is equal to one, we can find an optimal solution analytically. 
Let $M^Q=\{\rho_{\theta}|{\theta}\in\Theta\subset\bbr\}$ be a one-parameter quantum-state model. 
Then, the well-known property of the Fisher information results in the following inequalities. 
\begin{equation} 
J_{\theta}[\Pi]\le J^{\sld}_{\theta}[\rho_{\theta}]\quad \forall \Pi\in\cE, \label{1paraopt}
\end{equation}
where $J^{\sld}_{\theta}[\rho_{\theta}]$ is the {\it symmetric logarithmic derivative} (SLD) 
Fisher information about the parametric state $\rho_{\theta}$. 

To remind ourselves, the SDL Fisher information matrix about the mixed-state $\rho_{\bm{\theta}}$ is defied as follows. 
Consider a general $n$-parameter family of states $M^Q:=\{\rho_{\bm \theta}\,|\,{\bm \theta}\in\Theta\}$. 
The $i$th direction of SLD operator is defined by the solution of the operator equation 
$\del \rho_{\bm \theta}/\del \theta_i=(\rho_{\bm \theta} L_{{\bm \theta},i}+L_{{\bm \theta},i}\rho_{\bm \theta})/2$.  
The SLD Fisher information matrix about the model $\{\rho_{\bm \theta}\}$ is then defined by 
$J_{\bm \theta}^\sld[\rho_{\bm \theta}]=\big[ \tr{\rho_{\bm \theta}(L_{{\bm \theta},i}L_{{\bm \theta},j}+L_{{\bm \theta},j}L_{{\bm \theta},i})}/2\big]_{i,j}$. 
The SLD Fisher information is a quantum version of the Fisher information and is calculated 
solely by a given parametric quantum state. In the following, 
we denote it as $J_{\bm \theta}^{\sld}=J_{\bm \theta}^{\sld}[\rho_{\bm \theta}]$ for simplicity when no confusion arises. 

An optimal measurement attaining the above equality \eqref{1paraopt} is known.\cite{young,nagaoka87,bc94}. 
Hence, we can bound all possible Fisher information by the optimal one as \eqref{1paraopt}. 
This corresponds to the L\"owner optimal design and hence we can conclude that 
this is the optimal among all possible designs including the mixed strategy. 

\subsubsection{$c$-optimal design}\label{sec:c_opt}
In the literature, the $c$-optimal design for the quantum-state estimation is known, see for example, Chap.~7 of Ref.~\cite{ANbook}. 
\begin{theorem}\label{thm:copt}
Given an $n$-parameter model $M^Q$, for each $n$-dimensional (column) vector $\bm{c}=(c_1,c_2,\dots,c_n)^t\in\bbr^n$, 
the infimum of the MSE matrix in the direction of $\bm{c}$ is 
\be \label{eq:c-Opt}
\inf_{\e\in\cE} \bm{c}^t J_{\bm \theta}[\e]^{-1} \bm{c}
= \bm{c}^t ({J^\sld_{\bm \theta}})^{-1}\bm{c}.
\ee
An optimal measurement is given by a set of projectors about the operator:
\be
L_{{\bm \theta},\bm{c}}=\sum_{i,j=1}^n c_i J_{\bm \theta}^{\sld,ij} L_{{\bm \theta},j}, 
\ee
with $J_{\bm \theta}^{\sld,ij}$ the $i,j$ component of the inverse of the SLD Fisher information matrix and $L_{{\bm \theta},j}$ the SLD operator for the $i$th parameter $\theta_i$. 
\end{theorem}
This theorem provides an operational meaning of the SLD Fisher information matrix. 
In Sec.~5 of the review~\cite{suzuki2020quantum}, the detailed discussion on this Theorem was given in the context of the nuisance parameter problem. 

In general, the optimal design given in this theorem depends on the unknown parameter $\bm{\theta}$ as well as the choice of the known vector $\bm{c}$. 
Furthermore, the classical Fisher information matrix becomes singular, and hence it is the singular design problem. 
To circumvent the singular design problem, one should solve the refined optimization problem given by Eq.~\eqref{def:c-opt_singular}. 
Otherwise, an obtained optimized design describes purely mathematical one, which is useless. 
We illustrate this point by a simple example in the next subsection. 

\subsection{Local optimal design and singular design}\label{sec:local_singular}
In this subsection, we expand discussions on the issue of local optimal design and singular design, which were briefly presented in Sec.~\ref{sec:misc}. 

Consider a two-parameter qubit model given by
\begin{equation}
M^Q=\left\{\left.\rho_{\bm \theta}=\frac12\left(\begin{array}{cc}1& \theta_2\Exp{-\I \theta_1} \\
\theta_2\Exp{\I \theta_1} & 1\end{array}\right)\,\right|\,(\theta_1,\theta_2)\in[0,2\pi)\times(0,1)\right\}.
\end{equation}
The SLD Fisher information matrix of this model is 
\begin{equation}
J_{\bm{\theta}}^\sld=\left(\begin{array}{cc}\theta_2^2& 0 \\
0 &
\frac{1}{1-\theta_2^2} \end{array}\right)  . 
\end{equation}

When we are only interested in estimating the phase of this state $\theta_1$, 
whereas $\theta_2$ is treated as the nuisance parameter. 
The optimal design for the parameter of interest $\theta_1$ is obtained by the $c$-optimality with $\bm{c}=(1,0)^t$. 
Theorem \ref{thm:copt} provides an optimal projection measurement as
\begin{equation}
\Pi_{*}(\theta_1)=\left\{\frac12\left(\begin{array}{cc}1& \pm\I\Exp{-\I \theta_1} \\
\mp\I\Exp{\I \theta_1} & 1\end{array}\right)  \right\}. 
\end{equation}
Clearly, this measurement depends on the unknown parameter $\theta_1$, 
and hence it is a local optimal design. 
The classical Fisher information matrix of this optimal measurement is 
\begin{equation}\label{eq:ex_optFI}
J^*_{\bm{\theta}}:=J_{\bm{\theta}}[\Pi_{*}(\theta_1)]=\left(\begin{array}{cc}\theta_2^2& 0 \\
0 &0 \end{array}\right)  . 
\end{equation}
Thus, this optimal $\Pi_{*}(\theta_1)$ is the singular design in our terminology. 

First, let us discuss the issue of estimability discussed in Sec.~\ref{sec:misc}.  
The feasibility cone \eqref{def:feasibility_cone} for the parameter $\theta_1=\bm{c}^t\bm{\theta}$ is given by 
\begin{equation}
\cA(\bm{c})=\{A\in\bbr^{2\times2}|A>0\}\cup \{a{\bm{c}}{\bm{c}}^t=\left(\begin{array}{cc}a& 0 \\
0 &0 \end{array}\right)|a>0\}. 
\end{equation}
We see that the Fisher information matrix \eqref{eq:ex_optFI} for the $c$-optimal design $\Pi_{*}(\theta_1)$ is in this feasibility cone. Hence, $\theta_1$ is estimable by this optimal design. 
Next, we touch on the singular design problem. Define the set of all generalized inverse matrices of $J^*_{\bm{\theta}}$ by 
\begin{equation}
\rm{GI}(J^*_{\bm{\theta}}):=\{J^-\in\bbr^{2\times2}|J^*_{\bm{\theta}}J^-J^*_{\bm{\theta}}=J^*_{\bm{\theta}}\}.
\end{equation}
It is easy to obtain the following explicit expression. 
\begin{equation}
\rm{GI}(J^*_{\bm{\theta}})=\left\{\left. \left(\begin{array}{cc}(\theta_2)^{-2}& a \\
b &c \end{array}\right) \right|a,b,c\in\bbr\right\}. 
\end{equation}
Therefore, any generalized inverse of $J^*_{\bm{\theta}}$ attain the optimal value as 
\begin{equation}
\bm{c}^t J^{-} \bm{c}=\bm{c}^t ({J^\sld_{\bm \theta}})^{-1}\bm{c}, \ \forall J^-\in\rm{GI}(J^*_{\bm{\theta}}). 
\end{equation}
This suggests that there are other optimal design whose Fisher information matrix gives 
the same generalized inverse is in the set $\rm{GI}(J^*_{\bm{\theta}})$. 
However, one should only consider an optimal design lying on the feasibility cone, 
otherwise it only gives a meaningless design. 

Finally, we elaborate on local optimality of the design $\Pi_{*}(\theta_1)$. 
In reality, we can only perform an approximated optimal design with uncertainty in $\theta_1$ 
in the finite sample case. 
Upon using $\e_\delta:=\Pi_{*}(\theta_1+\delta)$ with an uncertainty $\delta$ in the knowledge about $\theta_1$, 
the Fisher information matrix and its (Moore-Penrose) generalized inverse is calculated as
\begin{align}
J_{\bm{\theta}}[\e_\delta]
&=\frac{1}{1-\theta_2^2\sin^2\delta}
\left(\begin{array}{c}\theta_2\cos\delta \\\sin\delta \end{array}\right) 
(\theta_2\cos\delta \ \sin\delta), \\
\big(J_{\bm{\theta}}[\e_\delta]\big)^{-1}
&=\frac{1-\theta_2^2\sin^2\delta}{ (\theta_2^2\cos^2\delta+\sin^2\delta)^2}
\left(\begin{array}{c}\theta_2\cos\delta \\\sin\delta \end{array}\right) 
(\theta_2\cos\delta \ \sin\delta).  
\end{align}
By evaluating the $(1,1)$ component of the generalized inverse, we obtain
\begin{equation}\label{eq:sec3-6}
\Psi_{\bm{c}}[\e_\delta]=\bm{c}^t \big(J_{\bm{\theta}}[\e_\delta]\big)^{-1} \bm{c}
= \frac{1-\theta_2^2\sin^2\delta}{ (\theta_2^2\cos^2\delta+\sin^2\delta)^2} \theta_2^2\cos^2\delta. 
\end{equation}
We can show that $\Psi_{\bm{c}}[\e_\delta]$ can be lower than its optimal value 
$\Psi_{\bm{c}}[\e_*]=(\theta_2)^{-2}$ for $\delta\neq0$. 
For example, consider the case of small $\delta$, then the Taylor series expansion gives
\begin{equation}
\Psi_{\bm{c}}[\e_\delta]\simeq\frac{1}{\theta_2^{2}}
- \frac{1}{\theta_2^{2}}\left[\theta_2^2+\frac{2}{\theta_2^2}- 1 \right]\delta^2 . 
\end{equation}
The second term is always negative. 
In fact, the optimal value for $\Psi_{\bm{c}}[\e_\delta]$ is zero, which is attained by 
a choice $\delta=\pi/2$ in Eq.~\eqref{eq:sec3-6}. 
To resolve this contradictory statement, we again need to impose the estimability condition. 
The parameter $\theta_1$ is estimable by the design $\e_\delta$, if and only if 
its Fisher information matrix is in the feasibility cone $\cA(\bm{c})$. 
This condition singles out the true optimal design with $\delta=0$.  

\section{Qubit model}\label{sec:qubit}
In this section, we consider a general qubit model; $M^Q=\{\rho_{\bm \theta}\in\sofctwo\,|\,{\bm \theta}\in\Theta\subset\bbr^n\}$. 
The single parameter case is solved in Sec.~\ref{sec:1para_opt}, and we consider two or three parameter models ($n=2,3$). 
For two- or three-parameter models, we can easily show that there cannot be 
L\"owner optimal design except for trivial cases. 

For the qubit model, a key observation is the following lemma (See Ref.~\cite{yamagata11} for the proof.). 
\begin{lemma}\label{lem1}
For a qubit model, $\{\rho_{\bm \theta}|{\bm \theta}\in\Theta\}$, the Fisher information matrix 
about given a measurement $\Pi$ can be expressed in the form as 
$J_{\bm \theta}[\Pi]=\sqrt{J_{\bm \theta}^{\sld}}J\sqrt{J_{\bm \theta}^{\sld}}$ where 
$J$ is some nonnegative-definite matrix satisfying the condition $\mathrm{Tr}\{J\}\le1$.
\end{lemma}

This immediately yields the following corollary, known as the Gill-Massar inequality.\cite{gm00} 
\begin{corollary} \label{cor_gm}
In a qubit system, the Fisher information matrix for any design $\xi$ satisfies,
\be
\Tr{\big(J_{\bm \theta}^{\sld}\big)^{-1}J(\xi)}\le1,
\ee
where the equality holds if and only if a measurement consists of rank-1 operators. 
\end{corollary}
Another important property is the following lemma. 
\begin{lemma}\label{lem2}
For any qubit model $M^Q=\{\rho_{\bm \theta}|{\bm \theta}\in\Theta\}$, 
let $T$ be a positive matrix, then the following optimization has the solution: 
\be
\max_{\e\in\cE}\Tr{J_{\bm \theta}[\e] T}=
\lambda_{\max}( J_{\bm \theta}^{\sld}T) ,
\ee
where the $\lambda_{\max}(A)$ denotes the maximum eigenvalue of the matrix $A$. 
\end{lemma}

\subsection{Fisher information region for a qubit model}\label{sec:qubit-FIregion}
We can apply Lemma \ref{lem1} to obtain the Fisher information region \eqref{Fisherset}, the set of all possible Fisher information matrices:
\begin{equation}\label{qFisherset}
\cJ(\Xi)=\{\sqrt{J_{\bm \theta}^{\sld}}J\sqrt{J_{\bm \theta}^{\sld}}\,|\,J\ge0,\mathrm{Tr}\{J\}\le1\}.
\end{equation}

To better understanding, we have an explicit construction of the Fisher information matrix 
based on the SLD operator. 
Let $L_{{\bm \theta},i}$ be the $i$th direction of the SLD operator. 
Given a unit vector $\bm{u}=(u^i)\in \bbr^n$ ($\bm{u}^t\bm{u}=|\bm{u}|^2=1$), 
performing a projection measurement about an observable, 
\be \label{sld_pvm}
L_{\bm{u}}:=\sum_{i,j}u^i (J_{\bm \theta}^\sld)^{-1/2}_{ij}L_{{\bm \theta},j},
\ee
yields the following form of the Fisher information matrix:
\begin{equation}\nonumber
J_{\bm \theta}[\Pi(L_{\bm u})]=\sqrt{J_{\bm \theta}^{\sld}}\bm{u} \bm{u}^t\sqrt{J_{\bm \theta}^{\sld}}, 
\end{equation}
which is rank-1. Consider a general experimental design $\e(n)=(\V{\p},\V{\e})$ for the $n$-parameter case 
of the form $\V{\p}=(\p_1,\dots,\p_n)\in\cP(n)$ and $\V{\e}=\V{\u}:=(\bm{u}_1,\dots,\bm{u}_n)$ 
$\stackrel{{1 \mathrm{to} 1}}{\leftrightarrow}
(\Pi(L_{\bm{u}_1}),\dots,\Pi(L_{\bm{u}_n})) $. (Note here that $\Pi(L_{\bm{u}}))$ is uniquely specified by a unit vector $\bm{u}$.) 
The Fisher information matrix for this design is
\begin{equation} \label{FisherE3}
J_{\bm \theta}[\e(n)]=\sum_{i=1}^n \p_i \sqrt{J_{\bm \theta}^{\sld}}\bm{u}_i \bm{u}_i^t\sqrt{J_{\bm \theta}^{\sld}}. 
\end{equation} 
The matrix $J=\sum_i \p_i \bm{u}_i \bm{u}_i^t$ satisfies $\mathrm{Tr}\{J\}=1$ and can span 
all possible nonnegative-definite matrices $J$ appearing in Eq.~\eqref{qFisherset}. 
Thus, we can set $n$ vectors $\bm{u}_1,\dots,\bm{u}_n$ to be orthogonal to each other 
to optimize the function $\Psi(J_{\bm \theta}[\e(n)])$ over $\V{\p}$ and $\V{\u}$. 
That is, $\V{\u}$ forms an orthonormal basis of $\bbr^n$. 
We can confirm that the design $\e(n)$ for the $n$-parameter case can achieve optimal design among all possible $\e(m)$ with $m\in\bbn$, that is, $\e_*(m_*)=\e_*(n)$ holds. 


Combining discussions above, we arrived at the following statement. 
\begin{proposition}
For any qubit model, let $\cJ(\Xi)$ be the Fisher information region for all possible designs, 
and denote by $\cJ(\cE_{\rm PVM})$ the Fisher information region set by a convex mixture of all possible projection measurements. 
Two Fisher information regions are identical for the quantum-state estimation problem: 
$\cJ(\cE_{\rm PVM})=\cJ(\Xi)$. 
\end{proposition}

\subsection{Analytical forms of optimal designs}
In this subsection, we study $A$-, $D$-, $E$-, and $\gamma$-optimal design. 
Each optimal design is constructed as randomized mixture of PVMs in consistent to the above proposition. 
We first derive the $\gamma$-optimal design, and then we list $A$-, $D$-, and $E$-designs. 

\subsubsection{$\gamma$-optimal design}
We finally construct the $\gamma$-optimal design for the qubit model. 
Its optimality function is $\Psi_\gamma[J_{\bm \theta}[\e(n)]]=\big(\frac{1}{n}\mathrm{Tr}\{J_{\bm \theta}[\e(n)]^{-\gamma}\} \big)^{1/\gamma}$ ($\gamma\in\bbr$). 
The result is given as follows. 
\begin{theorem}\label{theorem_gamma}
Given an $n$-parameter qubit model ($n=2,3$), an optimal design $\e_*(n)$ and 
the minimum $\gamma$-optimality function ($\gamma\neq0$) are given by
\begin{align}\nonumber
\min_{\e(n)}&\Psi_\gamma[J_{\bm \theta}[\e(n)]]= \frac{1}{n^\gamma} 
\big(\mathrm{Tr}\{(J_{\bm \theta}^{\sld})^{-\frac{\gamma}{\gamma+1}}\} \big)^{\frac{\gamma+1}{\gamma}},\\ \nonumber
\V{\p}_*&=(\p_i) \mbox{ with } \p_i=(\lambda_i^\sld)^{-\frac{\gamma}{\gamma+1}}/\sum_j (\lambda_j^\sld)^{-\frac{\gamma}{\gamma+1}} ,\\ \label{gamma-opt}
\V{u}_*&=(\bm{u}_i^\sld),
\end{align}
where $\lambda_i^\sld$ and $\bm{u}_i^\sld$ are the eigenvalues and eigenvectors of the SLD Fisher information matrix. 
The necessary and sufficient condition for the optimal design $\e_*$ is that 
the Fisher information matrix for $\e_*$ satisfies 
\begin{equation}\label{eq:thm-gamma}
J_{\rm \theta}[\e_*]=\frac{1}{\mathrm{Tr}\left\{(J_{\bm \theta}^{\sld})^{-\frac{\gamma}{\gamma+1}}\right\}}(J_{\bm \theta}^{\sld})^{\frac{1}{\gamma+1}}. 
\end{equation}
\end{theorem}

\begin{proof}
We extend the proof used in Ref.~\cite{yamagata11}. 
First note that it suffices to find the minimum for $\mathrm{Tr}\{J_{\bm \theta}[\e(n)]^{-\gamma}\}$. 
Consider the functional of $n\times n$ positive matrix $J>0$:
\begin{equation}
f(J):=\Tr{ (S J^{-1}S)^\gamma}+\lambda \left(\Tr{J}-1\right), 
\end{equation}
where $S=(J_{\bm \theta}^{\sld})^{-1/2}$ and $\lambda$ is the Lagrange multiplier. 
Taking a variation with respect to $J$ gives 
\begin{align*}
\delta f(J)&=\Tr{-(SJ^{-1}S)^\gamma\gamma(S^{-1}\delta JS^{-1})(S^{-1}JS^{-1})^{\gamma-1}(SJ^{-1}S)^\gamma}+\lambda\Tr{\delta J}\\
&=-\Tr{\left[\gamma S^{-1}(SJ^{-1}S)^{\gamma+1}S^{-1}-\lambda I  \right] \delta J}. 
\end{align*}
Therefore, the stationary condition yields the relation. 
\begin{align*}
\gamma S^{-1}(SJ_*^{-1}S)^{\gamma+1}S^{-1}-\lambda I=0
&\Lra (SJ_*^{-1}S)^{\gamma+1}=\frac{\lambda}{\gamma}S^2\\
&\Lra J_*=\left(\frac{\lambda}{\gamma}\right)^{\frac{1}{\gamma+1}} (J_{\bm \theta}^{\sld})^{-\frac{\gamma}{\gamma+1}}. 
\end{align*}
The condition $\Tr{J_*}=1$ determines $\lambda$ as 
\[
\lambda=\gamma \left(\mathrm{Tr}\left\{(J_{\bm \theta}^{\sld})^{-\frac{\gamma}{\gamma+1}}\right\}\right)^{\gamma+1}, 
\]
and the Fisher information matrix 
for the optimal design $\e_*$ is obtained as 
\begin{align*}
J_{\rm \theta}[\e_*]&=\sqrt{J_{\bm \theta}^{\sld}}J_*\sqrt{J_{\bm \theta}^{\sld}}\\
&=\frac{1}{\mathrm{Tr}\left\{(J_{\bm \theta}^{\sld})^{-\frac{\gamma}{\gamma+1}}\right\}}(J_{\bm \theta}^{\sld})^{\frac{1}{\gamma+1}}.
\end{align*}
This is equivalent to the condition \eqref{eq:thm-gamma}. 
This expression immediately gives expression for $\Psi_\gamma[J_{\bm \theta}[\e_*(n)]]$ in the theorem. 
To find an optimal design, we can solve $J_*=\sum_i \p_i \bm{u}_i \bm{u}_i^t$. 
It is straightforward to check the optimal design $\e_*(n)$ given in the theorem satisfies 
this relation. 
\end{proof}

\subsubsection{$A$-optimal design}
The $A$-optimal design for the qubit model is known. 
The Nagaoka bound corresponds to the case of $n=2$,\cite{nagaoka89} 
and Hayashi-Gill-Massar bound is identical to $n=3$.\cite{hayashi97,gm00} 
Yamagata gave a unified treatment for the qubit case as follows.\cite{yamagata11} 
\begin{align}\nonumber
\min_{\e(n)}&\Psi_A[J_{\bm \theta}[\e(n)]]= \frac{1}{n} 
\big(\mathrm{Tr}\{(J_{\bm \theta}^{\sld})^{-1/2}\} \big)^2,\\ \nonumber
\V{\p}_*&=(\p_i) \mbox{ with } \p_i=(\lambda_i^\sld)^{-1/2}
/\sum_j (\lambda_j^\sld)^{-1/2} ,\\
\V{\u}_*&=(\bm{u}_i^\sld),
\end{align}
where $\lambda_i^\sld$ and $\bm{u}_i^\sld$ are the eigenvalues and eigenvectors of the SLD Fisher information matrix $J_{\bm \theta}^{\sld}$.

\subsubsection{$D$-optimal design}
Let us discuss the $D$-optimal design. 
Since the Fisher information matrix $J_{\bm \theta}[\Pi]$ is expressed as in Eq.~\eqref{FisherE3}, 
the minimization of determinant of $J_{\bm \theta}[\e(n)]^{-1}$ is 
equivalent to maximize the value $\prod_{i}^n \p_i$. 
It is straightforward to see that $\V{\p}_*=(1/n,\dots,1/n)$ is the optimal choice 
for the $D$-optimal design, and we have 
\be
\Psi_D(J_{\bm \theta}[\e_*(n)])=n^{-n}\Det{{J_{\bm \theta}^{\sld}}^{-1}}. 
\ee
Furthermore, an optimal set of projection measurements is 
specified by arbitrary set of orthonormal vectors $\V{\u}_*=\{\bm{u}_i\}$ through 
expression \eqref{sld_pvm}.

\subsubsection{$E$-optimal design}
We next give the $E$-optimal design for the qubit model. 
As we remarked earlier, we only consider the full-rank Fisher information. 
Otherwise, any singular design cannot be used to estimate all parameters. 
An optimal measurement is again a set of measurements about the directions of the SLD operators 
as in Theorem \ref{theorem_gamma}. 
This then leads to the following minimization:
\begin{align*} 
\Psi_E(J_{\bm \theta}[\e_*(n)])&=\min_{\bm{\p}}\max\left\{\left(\p_i\lambda(J_{\bm \theta}^\sld)\right)^{-1}\right\}\\
&=\left(\max_{\bm{\p}}\min\left\{\p_i\lambda(J_{\bm \theta}^\sld)\right\}\right)^{-1}.
\end{align*}
The optimal relative frequency for the $E$-optimal design instead takes the form: 
\be
\p^*_i=\frac{(\lambda^\sld_i)^{-1}}{\sum_{j=1}^n (\lambda^\sld_j)^{-1}},
\ee 
where $\lambda^{\sld}_i$ are the eigenvalues of the SLD Fisher information matrix. 
The minimum value of the maximum eigenvalue of the Fisher information matrix is 
given by 
\be
\Psi_E(J_{\bm \theta}[\e_*(n)])=\Tr{{J_{\bm \theta}^{\sld}}^{-1}}. 
\ee

\section{Quantum equivalence theorem for a qubit system}\label{sec4}
In this section, we prove a quantum version of equivalence theorem. 
Combining the results regarding the qubit model yields the following theorem. 
\begin{theorem}
For any qubit model, the following optimization problems are equivalent. 
\begin{align*}
\mathrm{1)}&\quad \min_{\xi\in\Xi}\Det{ J(\xi)^{-1} },\\
\mathrm{2)}&\quad \min_{\xi\in\Xi} \Tr{J_{\bm \theta}^{\sld} J(\xi)^{-1} }. 
\end{align*}
that is the $D$-optimal design coincides with the $A$-optimal design with the weight matrix $J_{\bm \theta}^{\sld}$. 
\end{theorem}
\begin{proof}
%
Let us consider an alternative expression of the $D$-optimality function, 
$\Psi_D(J)= \log\big[\Det{J^{-1}}\big]=-\log [\Det{J}] $. 
The sensitivity function is given by $\varphi(\e,\xi)=\Tr{J_{\bm \theta}[\e] J(\xi)^{-1}}$. 
From Theorem \ref{thm_equiv2}, we have the following equivalence:
\[
\xi_*=\min_{\xi\in\Xi}\Psi_D(J(\xi))\Lra \max_{\e\in\cE}\Tr{J_{\bm \theta}[\e] J(\xi_*)^{-1}}= n. 
\]
The maximization problem is solved by Lemma \ref{lem2} to get the condition: 
\be\label{Dopt_cond}
\lambda_{\max}\left( \hat{J}(\xi_*)^{-1} \right)=n, 
\ee
where $ \hat{J}(\xi_*):=(J_{\bm \theta}^{\sld})^{-1/2}J(\xi_*)(J_{\bm \theta}^{\sld})^{-1/2}$. 
Next, from the sensitivity function for the $A$-optimality function 
with the weight matrix $J_{\bm \theta}^\sld$, 
we find $\xi_*$ is $A$-optimal if and only if 
\begin{align} \nonumber
&\max_{\e\in\cE} \Tr{J_{\bm \theta}^\sld J(\xi_*)^{-1} J_{\bm \theta}[\e]J(\xi_*)^{-1}}= \Tr{J_{\bm \theta}^{\sld} J(\xi_*)^{-1}}\\
&\Lra \lambda_{\max}\left( \hat{J}(\xi_*)^{-2} \right)=\Tr{\hat{J}(\xi_*)^{-1}  },
\end{align}
where Lemma \ref{lem2} was used. 
Let $j_1\ge j_2\ge\dots\ge j_n(>0)$ be the eigenvalues of $\hat{J}(\xi_*)$, 
then Corollary \ref{cor_gm} states $\sum_k j_k\le1$. 
With this notation, the $D$-optimality condition is equivalent to $j_n^{-1}=n$. 
This is also equivalent to $j_k=1/n$ for all $k$ due to 
the constraint $\sum_k j_k\le1$. 
Finally, the $A$-optimality condition is expressed as 
$j_n^{-2}=\sum_k j_k^{-1}$. This then implies $j_k=1/n$ for all $k$. 
This completes the proof. 
\end{proof}

\section{Comparison of optimal designs}\label{sec5}
In this section, we compare optimal designs for 
$A$-, $D$-, and $E$-optimality criteria. We denote these optimal designs 
by ${\e}_A$, ${\e}_D$, and ${\e}_E$, respectively. 
As an another reference, we consider the so called the standard tomography ${\e}_{ST}$. 
This is defined by the design $\e_{ST}=(\V{\p}_{ST},\V{\e}_{ST})$ with $\V{\p}_{ST}=(1/3,1/3,1/3)$ and 
$\V{\e}_{ST}=(\Pi^{(1)},\Pi^{(2)},\Pi^{(3)})$. 
Here, $\Pi^{(k)}$ ($k=1,2,3$) are the projection measurements 
about $k$th Pauli matrix $\sigma_k$.  

We first list Fisher information matrices for these designs. 
\begin{align}
J_{A}&:=J_{\bm \theta}[{\e}_A]=\frac{1}{\Tr{{(J_{\bm \theta}^{\sld})}^{-1/2}}} {(J_{\bm \theta}^{\sld})}^{1/2}, \label{JoptA0} \\
J_{D}&:=J_{\bm \theta}[{\e}_D]=\frac{1}{n} J_{\bm \theta}^{\sld},\label{JoptD0}\\
J_{E}&:=J_{\bm \theta}[{\e}_E]=\frac{1}{\Tr{{(J_{\bm \theta}^{\sld}}^{-1})}} \openone_n,\label{JoptE0} \\
J_{ST}&:=J_\theta[{\e}_{ST}]
=\left[\sum_{k=1,2,3} \frac{1}{1-s_{{\bm \theta},k}^2}  \frac{\del s_{{\bm \theta},k}}{\del \theta_i}\frac{\del s_{{\bm \theta},k}}{\del \theta_j} \right], \label{JoptSt0}
\end{align}
where $n$ is the number of parameters and $\openone_n$ denotes the $n\times n$ identity matrix. 
For the Fisher information matrix of the standard tomography, 
we use the Bloch vector representation of the state, $s_{{\bm \theta},j}=\tr{\rho_{{\bm \theta}}\sigma_j}$ 
with $\sigma_j$ ($j=1,2,3$) the Pauli matrices. 
To see the general structure, we also show the Fisher information matrix for the $\gamma$-optimal design ${\e}_\gamma$: 
\begin{equation}\label{Jopt_gamma}
J_{\gamma}:=J_{\bm \theta}[{\e}_\gamma]=\frac{1}{\Tr{\left(J_{\bm \theta}^{\sld}\right)^{-\frac{\gamma}{\gamma+1}}}} 
\left(J_{\bm \theta}^{\sld}\right)^{\frac{1}{\gamma+1}} .
\end{equation}
Let us observe from this result that the structure is very different for each optimal Fisher information matrix. Explicitly, the eigenvalues of $J_{A}$, $J_{D}$, and $J_{E}$ are all different in general. 

As a concrete example, we consider the standard parametrization of the general qubit state with the Stokes parameters. 
Its model is given by
\be\label{qbit_model}
M^Q=\left\{\left.\rho_{\bm \theta}=\frac12\left(\begin{array}{cc}1+\theta_3 & \theta_1-\I\theta_2 \\
\theta_1+\I\theta_2  & 1-\theta_3\end{array}\right)\,\right|\,{\bm \theta}\in\Theta\right\}.
\ee
Here, $\Theta=\left\{{\bm \theta}\in\bbr^3\,\big|\,|{\bm \theta}|^2=\sum_i(\theta_i)^2<1  \right\}$. 
For this model, the SLD Fisher information matrix can be computed, 
and its inverse is 
\[{(J_{\bm \theta}^{\sld})}^{-1}=\openone_3-\bm{\theta}\bm{\theta}^t, \]
where $\bm{\theta}:= (\theta_1,\theta_2,\theta_3)^t$ denotes the column vector. 
We list the inverse matrices of the Fisher information for each optimal design: 
\begin{align} \label{JoptA}
J_{A}^{-1}&=(2+\sqrt{1-|\bm{\theta}|^2}) \left[ \openone_3 +\frac{\sqrt{1-|\bm{\theta}|^2}-1}{|\bm{\theta}|^2} \bm{\theta}\bm{\theta}^t \right], \\
J_{D}^{-1}&=3(\openone_3- \bm{\theta}\bm{\theta}^t ),\label{JoptD} \\
J_{E}^{-1}&= \big(3-|\bm{\theta}|^2\big)\openone_3,\label{JoptE} \\
J_{ST}^{-1}&  = 3\, \mathrm{diag.}\left( 1-(\theta_1)^2,\,1-(\theta_2)^2,\,1-(\theta_3)^2 \right). \label{JoptSt}  
\end{align}

\subsection{$A$-optimality}
Let us consider the $A$-optimality function $\Psi_A(J)=\Tr{J^{-1}}$. 
First, values of the $A$-optimality function for optimal designs are 
\begin{align*}
\Psi_A(J_A)&= \left(  \Tr{(J_{\bm \theta}^{\sld})^{-1/2}}\right)^2,\\
\Psi_A(J_D)&=n  \Tr{(J_{\bm \theta}^{\sld})^{-1}},\\
\Psi_A(J_E)&=n  \Tr{(J_{\bm \theta}^{\sld})^{-1}}=\Psi_A(J_D),
\end{align*}
Here we omit expression for the standard tomography design, since it is rather lengthy. 
From above results, we immediately see that $\e_D$ and $\e_E$ perform exactly same in terms of the $A$-optimality. 

We next consider model \eqref{qbit_model}. 
The results including the standard tomography design are
\begin{align*}
\Psi_A(J_A)&= (2+\sqrt{1-|\bm{\theta}|^2})^2,\\
\Psi_A(J_D)&=\Psi_A(J_E)=\Psi_A(J_{ST})=3 \big(3-|\bm{\theta}|^2\big). 
\end{align*} 
Interestingly, $\Psi_A$ takes the same values for $\e_D,\e_E,\e_{ST}$. 

As a normalized version of these values, we compare the efficiency, $\eta_A[\e]=\Psi_A(J_A)/\Psi_A(J[\e])$, defined in Sec.~\ref{sec:contDoE}. 
By definition, $\eta_A[\e_A]=1$ and others are
\[
\eta_A(\e_D)=\eta_A(\e_E)=\eta_A(\e_{ST})=\frac{(2+\sqrt{1-|\bm{\theta}|^2})^2} {3 \big(3-|\bm{\theta}|^2\big)}. 
\]
In Fig.~\ref{fig1}, we plot efficiency functions $\eta_A[\e_A]=1$ (Black solid curve)
and $\eta_A(\e_D)=\eta_A(\e_E)=\eta_A(\e_{ST})$ (Gray solid curve) as a function of $|{\bm \theta}|^2$. 
It is clear that $\eta_A(\e_D)=\eta_A(\e_E)=\eta_A(\e_{ST})$ is a monotonically 
decreasing function of $|\bm{\theta}|^2=\sum_i(\theta_i)^2$. 
The infimum is given by the pure-state limit $|\bm{\theta}|^2\to1$, 
whose value is $2/3$. For small values of $|\bm{\theta}|^2$, on the other hand, it becomes close to one. 
This means that there is no significant difference among different optimal designs when a state is closed to the completely mixed state. 
\begin{figure}[htbp]
\centering
\includegraphics[width=0.475\textwidth,keepaspectratio,clip]{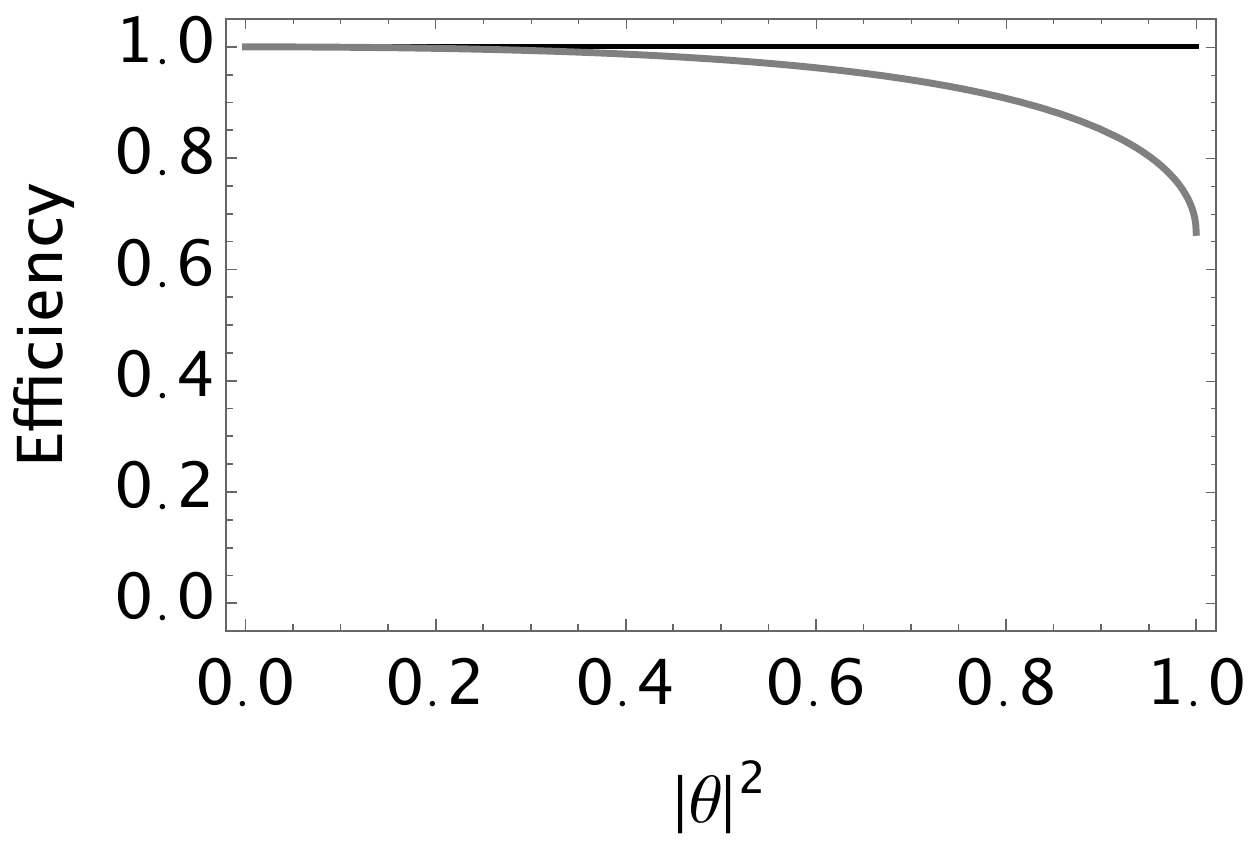}
\label{fig1}
\caption{Efficiency functions for the $A$-optimality criterion.}
\end{figure}

\subsection{$D$-optimality}
Let us consider the $D$-optimality function $\Psi_D(J)=\Det{J^{-1}}$. 
Values of the $D$-optimality function for optimal designs are 
\begin{align*}
\Psi_D(J_A)&= \Tr{(J_{\bm \theta}^{\sld})^{-1/2}} \Det{(J_{\bm \theta}^{\sld})^{-1/2}},\\
\Psi_D(J_D)&=n^n  \Det{(J_{\bm \theta}^{\sld})^{-1}},\\
\Psi_D(J_E)&= \left(\Tr{(J_{\bm \theta}^{\sld})^{-1}}\right)^n. 
\end{align*}

For model \eqref{qbit_model}, the results are
\begin{align*}
\Psi_D(J_A)&= (2+\sqrt{1-|\bm{\theta}|^2})\sqrt{1-|\bm{\theta}|^2} ,\\
\Psi_D(J_D)&=3^3 \big( 1-|\bm{\theta}|^2\big), \\
\Psi_D(J_E)&= \big( 3-|\bm{\theta}|^2\big)^3,\\
\Psi_D(J_{ST})&=3^3 \prod_{k=1,2,3} \big( 1-(\theta_k)^2\big). 
\end{align*} 
Efficiencies are calculated as
\begin{align*}
\eta_D(\e_A)&=\frac{3^3 \sqrt{1-|\bm{\theta}|^2}}{ \left(2+\sqrt{1-|\bm{\theta}|^2}\right)^3},\\
\eta_D(\e_E)&=\frac{3^3 \big(1-|\bm{\theta}|^2\big)}{ \big(3-|\bm{\theta}|^2\big)^3},\\
\eta_D(\e_{ST})&=\big(1-|\bm{\theta}|^2\big) \prod_{k=1,2,3} \big( 1-(\theta_k)^2\big)^{-1}.
\end{align*}
When compared to the $A$-optimal case, the performance of the standard tomography is not rotationally symmetric. 
To be specific, it efficiency $\eta_D(\e_{ST})$ explicitly depends on the direction of the Bloch vector. 

In Fig.~\ref{fig2}, we plot four efficiency functions 
$\eta_D[\e_A]$ (Black solid curve), $\eta_D(\e_D)=1$ (Dotted curve), $\eta_D(\e_E)$ (Dashed curve), $\eta_D(\e_{ST})$ (Gray solid curve) as a function of $|{\bm \theta}|^2$. 
To produce these figures, we fix a particular direction of the Bloch vector given by 
$(\sin\theta_0\cos\phi_0,\sin\theta_0\sin\phi_0,\cos\theta_0)$ and then we change the square of the length $|{\bm \theta}|^2$. 
In the left plot, we choose $\theta_0=\pi/16,\phi_0=\pi/4$. 
Another choice $\theta_0=\pi/4,\phi_0=\pi/4$ is made for the right plot. 

From Fig.~\ref{fig2}, the following relation is expected to hold. 
\[
1=\eta_D(\e_D)\ge\eta_D(\e_A),\eta_D(\e_{ST})\ge\eta_D(\e_E).  
\] 
We now show this ordering. 
The first inequality holds by definition. 
To verify the last inequality, we first show $\eta_D(\e_A)>\eta_D(\e_E)$ 
$\Leftrightarrow$ $\Psi_D(J_A)\le\Psi_D(J_E)$ for all $|{\bm \theta}|^2>0$. 
This can be done by analyzing $\Psi_D(J_E)-\Psi_D(J_A)$ as a function of $|{\bm \theta}|^2$. 
The other relation $\eta_D(\e_{ST})\ge\eta_D(\e_E)$ $\Leftrightarrow$ $\Psi_D(J_{ST})\le\Psi_D(J_E)$ obeys from the inequality of arithmetic and geometric means.  
From Fig.~\ref{fig2}, we see that there is no ordering between 
$\eta_D(\e_A)$ and $\eta_D(\e_{ST})$. 

Next, we note that $\eta_D(\e_A),\eta_D(\e_E),\eta_D(\e_{ST})$ become zero as 
$|\bm{\theta}|^2$ approaches one (The pure-state limit). 
This indicates that designs $\e_A, \e_E, \e_{ST}$ become completely useless in terms of $D$-optimality. 
We elaborate on this in Sec.~\ref{sec:discussion}. 

\begin{figure}[htbp]
\centering
\includegraphics[width=0.475\textwidth,keepaspectratio,clip]{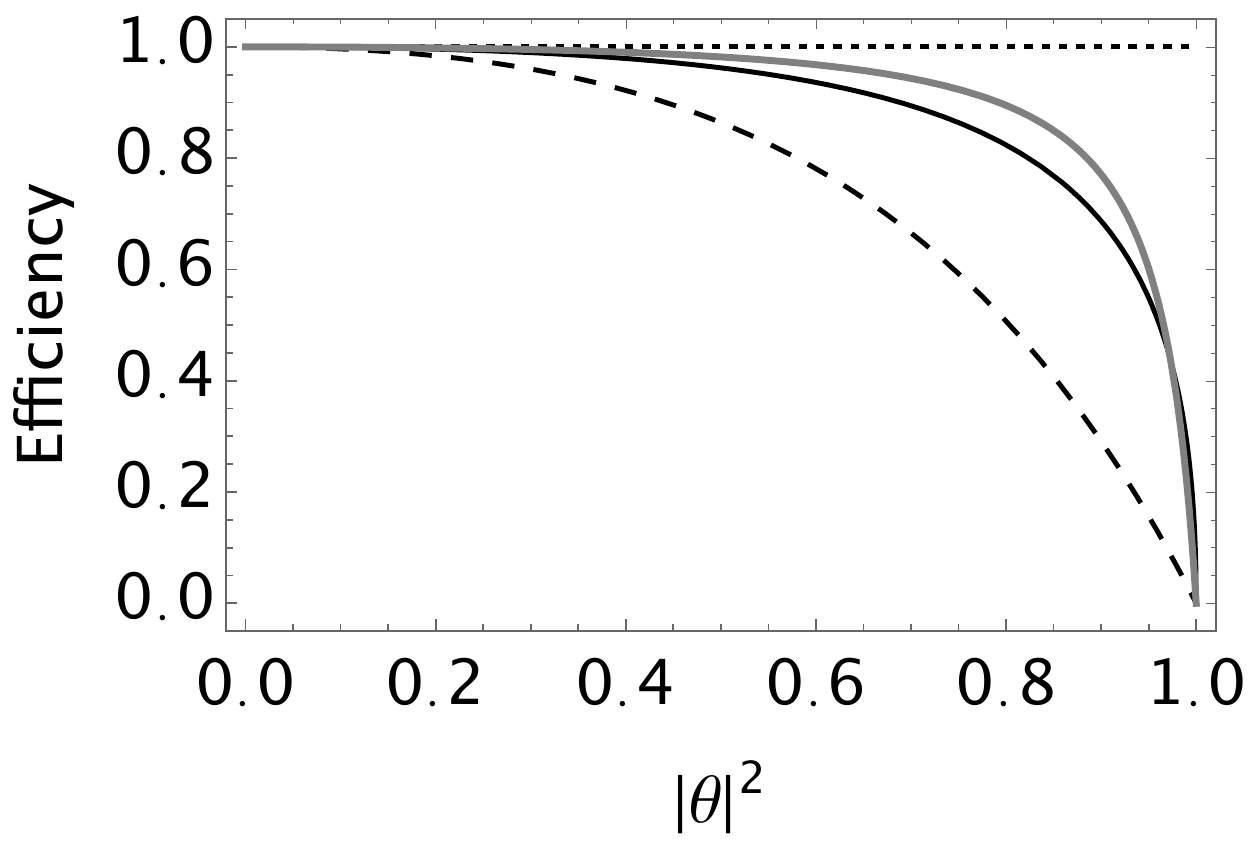}
\includegraphics[width=0.475\textwidth,keepaspectratio,clip]{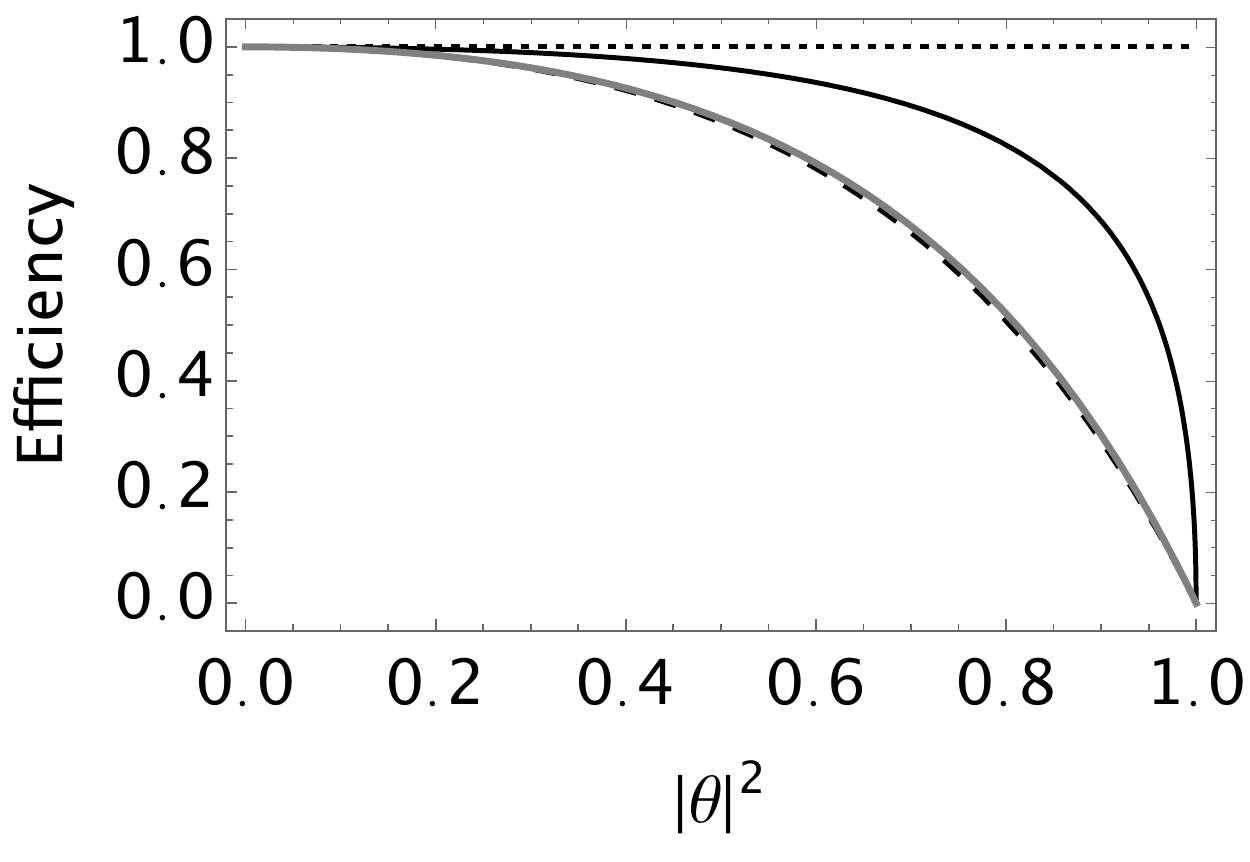}
\label{fig2}
\caption{Efficiency functions for the $D$-optimality criterion.}
\end{figure}


\subsection{$E$-optimality}
Let us consider the $E$-optimality function $\Psi_E(J)= \lambda_{\max}\{J^{-1}\}$. 
Values of the $E$-optimality function for optimal designs are 
\begin{align*}
\Psi_E(J_A)&= \Tr{(J_{\bm \theta}^{\sld})^{-1/2}} \lambda_{max}\{(J_{\bm \theta}^{\sld})^{-1/2}\},\\
\Psi_E(J_D)&=n  \lambda_{\max}\{(J_{\bm \theta}^{\sld})^{-1}\},\\
\Psi_E(J_E)&= \Tr{(J_{\bm \theta}^{\sld})^{-1}}. 
\end{align*}

For model \eqref{qbit_model}, we have
\begin{align*}
\Psi_E(J_A)&= 2+ \sqrt{1-|\bm{\theta}|^2},\\
\Psi_E(J_D)&=3,\\
\Psi_E(J_E)&=3-|\bm{\theta}|^2,\\
\Psi_E(J_{ST})&=3 (1-\min\{(\theta_i)^2\}). 
\end{align*}
Efficiencies are obtained as
\begin{align*}
\eta_E(J_A)&=\frac{3-|\bm{\theta}|^2}{ 2+\sqrt{1-|\bm{\theta}|^2}} ,\\
\eta_E(J_D)&=\frac13\big( 3-|\bm{\theta}|^2\big),\\
\eta_E(J_{ST})&= \frac{3-|\bm{\theta}|^2}{3 (1-\min\{(\theta_i)^2\})} . 
\end{align*} 

\begin{figure}[htbp]
\centering
\includegraphics[width=0.475\textwidth,keepaspectratio,clip]{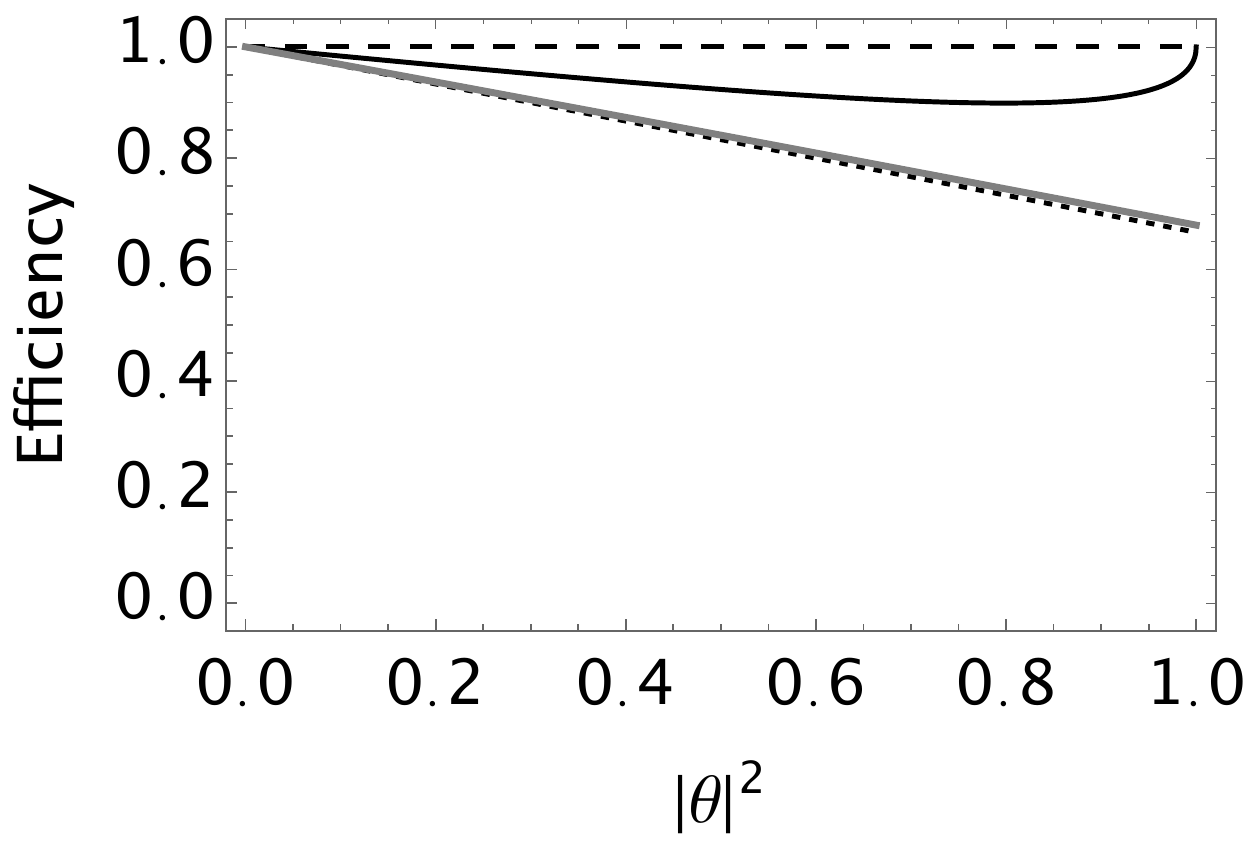}
\includegraphics[width=0.475\textwidth,keepaspectratio,clip]{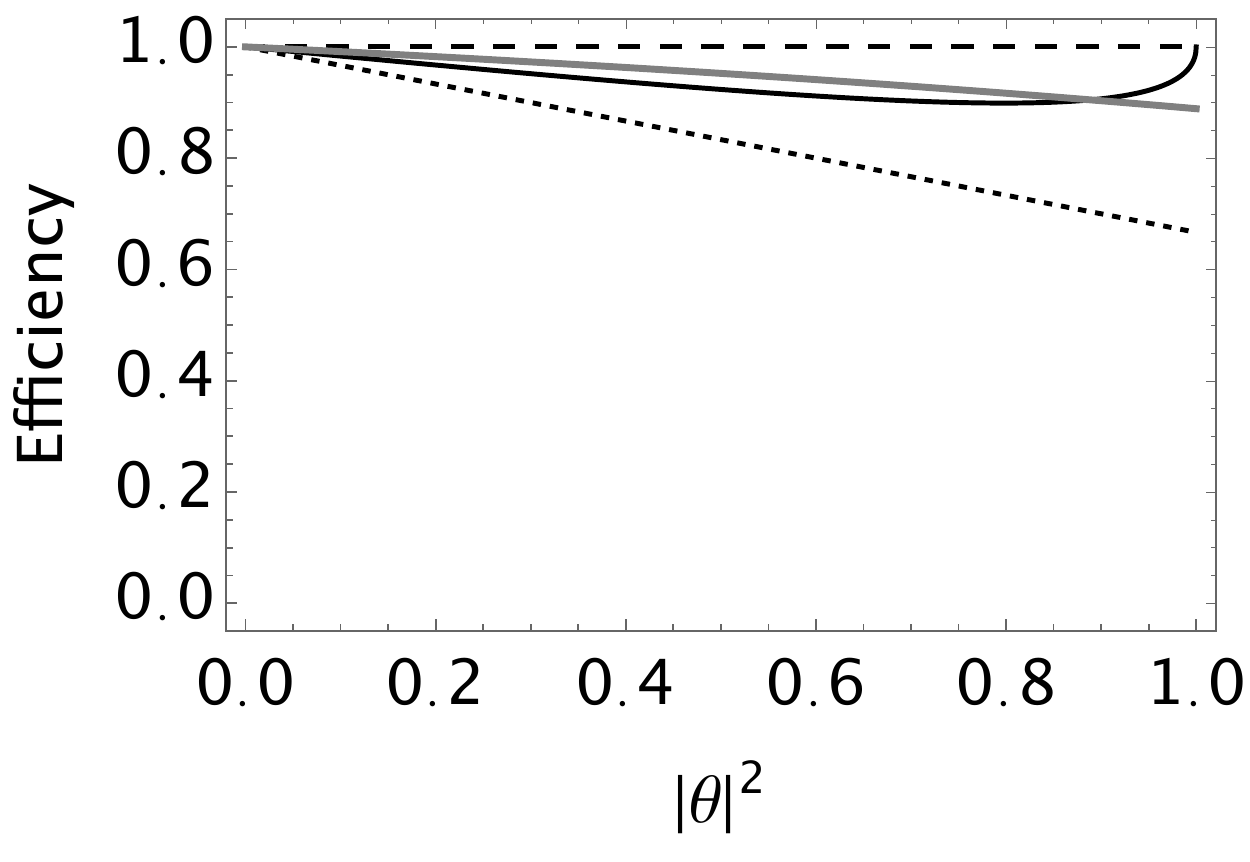}
\label{fig3}
\caption{Efficiency functions for the $E$-optimality criterion.}
\end{figure}

In Fig.~\ref{fig3}, we plot four efficiency functions 
$\eta_E[\e_A]$ (Black solid curve), $\eta_E(\e_D)$ (Dotted curve), $\eta_E(\e_E)=1$ (Dashed curve), 
$\eta_E(\e_{ST})$ (Gray solid curve) as a function of $|{\bm \theta}|^2$. 
As in Fig.~\ref{fig2}, we choose two particular directions of the Bloch vector: $\theta_0=\pi/16,\phi_0=\pi/4$ for the right plot and $\theta_0=\pi/4,\phi_0=\pi/4$ for the right plot. 
For efficiency about the $E$-optimality, we have the following ordering relation:
\[
1=\eta_E(\e_E)\ge\eta_E(\e_A),\eta_E(\e_{ST}) \ge \eta_E(\e_D).  
\] 
The relation $\Psi_E(J_A)\le  \Psi_E(J_D)$ can be shown by a straightforward exercise. 
The other inequality $\Psi_E(J_{ST})\le  \Psi_E(J_D)$ also holds trivially. 
Figure \ref{fig3} shows that there is no ordering between 
$\eta_E(\e_A)$ and $\eta_E(\e_{ST})$. 

Another interesting characteristic is that 
$\eta_E(\e_A)$ becomes one as $|\bm{\theta}|^2$ approaches one. 
Also, $\eta_E(\e_D)$ and $\eta_E(\e_{ST})$ do not vanish at the boundary $|\bm{\theta}|^2=1$. 

\subsection{Discussions}\label{sec:discussion}
The first observation in our study is that each optimality criterion defines a different optimal design, whose characteristics can be very different. 
Although this is clear, we explicitly demonstrate this fact for the popular optimality criteria. 
This point is illustrated by expressions for the Fisher information matrices 
\eqref{JoptA0}, \eqref{JoptD0}, \eqref{JoptE0}, and \eqref{Jopt_gamma}. 
In the following, we list more specific findings. 

The Fisher information matrix  for the design for the standard tomography \eqref{JoptSt} shows asymmetry in the Bloch vector representation. 
However, it becomes rotationally invariant for the $A$-optimality function by taking a trace with a unit weight matrix. This point was also demonstrated by Yamagata\cite{yamagata11} deriving the necessary and sufficient condition for the weight matrix such that the standard tomography coincides with the $A$-optimal design. See also a related work \cite{PhysRevA.90.012115} 
When the standard tomography is evaluated for other optimality criteria, we see that it is not the worst design among other optimal designs. 
This might be another justification of adopting the standard tomography in practice.  

Next, we analyze $A$-optimal design. 
By evaluating efficiency functions for other optimality criteria, 
we see that it is relatively stable. 
In particular, it behaves well for the $E$-optimality when compared with the $D$-optimal design and the standard tomography. 
One of the reasons behind this observation is that $A$-optimality with a unit weight matrix optimizes three parameters equal footings. 
Thus, we expect that it should perform well on average. 

The $D$-optimal design is known to be one of the most popular criteria in the classical theory of optimal DoE. 
However, its applicability in the quantum case needs further justification. 
Figure 2 shows that other optimal designs as well as the standard tomography become less efficient when the model becomes pure. 
This is because this optimal criterion concerns the product of eigenvalues of the Fisher information matrix, and thus it is sensitive to the small numbers. 
In particular, the $D$-optimality function $\Psi_D(J_D)$ for the $D$-optimal design 
vanishes in the pure-state limit. Therefore, efficiency function $\eta_D$ also vanishes unless $\Psi_D(J)$ cancels each other. 
The classical Fisher information matrix of this $D$-optimal design \eqref{JoptD} is proportional to the SLD Fisher information matrix. In literature\cite{LFGKC,ZH18}, the existence of such POVMs is not trivial in general, and it has been the subject known as the Fisher-symmetric informationally complete measurement. Our result on the $D$-optimal design is thus related. 
From this observation, $D$-optimality for the higher dimensional case is worth a further study. 

Last, let us make a brief comment on the $E$-optimal design. 
This optimality is related to the philosophy of the min-max strategy: 
One tries to avoid the worst case value of the MSE matrix. 
Interestingly, the classical Fisher information matrix for the $E$-optimal design is proportional to the identity matrix as seen in Eq.~\eqref{JoptE0}. 
This result exhibits the maximum symmetry for the Fisher information matrix. 
This optimal design is not so common in the quantum domain so far. 
It should play an important role when one wishes to guarantee the best 
estimate for the smallest eigenvalue of the Fisher information matrix.  

\section{Summary and outlook}\label{sec6}
In summary, we have formulated the problem of quantum-state estimation 
problem in the framework of optimal design of experiments (DoE). 
This formulation shows that the problem at hand is a usual statistical optimization problem 
except for the fact that quantities are represented by non-negative complex matrices. 
We have solved the qubit case analytically deriving popular optimal designs. 
A quantum version of the equivalence theorem is also proven in the qubit case. 
Another important finding of this paper is a comparison among the popular optimal designs: $A$-, $D$-, and $E$-optimal designs. 
In particular, we have shown that some of the optimal designs do not perform well 
for the other choice of optimality criterion. Although this is likely to happen in general, 
we have explicitly demonstrated it for the standard parametrization of qubit states. 

An important future work is to apply our formulation to various physically important 
problems and to find a good experimental setup by solving the optimization problem numerically. 
There are several open problems along the line of this research. 
First, to develop an efficient optimization algorithm for finding an optimal design. 
Second, generalization of the equivalence theorem to higher dimensional systems. 
Third, the singular design problem that is common in finding an optimal design in the presence of nuisance parameters.\cite{suzuki2020nuisance,tsang2020semipara,suzuki2020quantum} 
Classical theory of optimal DoE is a rich and mature subject in classical statistics. 
There are many unexplored subjects of DoE in the quantum case, 
which would be of great importance in any quantum information processing, 
such as a sequential design, block design, Bayesian design, minimax design, robust design, model discrimination, to list a few.

\section*{Acknowledgement}
The work is partly supported by JSPS KAKENHI Grant Number JP17K05571 and the FY2020 UEC Research Support Program, the University of Electro-Communications. 
He would like to thank Prof.~Hui Khoon Ng for invaluable discussions and her kind hospitality at Centre for Quantum Technologies in Singapore where part of this work was done.


\begin{thebibliography}{51}

\bibitem{fisher1960design}
R.~A. Fisher {\em et~al.}, {\em The design of experiments.}, no.~7th Ed (Oliver
  and Boyd. London and Edinburgh, 1960).

\bibitem{fedorov}
V.~V. Fedorov, {\em Theory of optimal experiments} (Academic Press, 1972).

\bibitem{pukelsheim}
F.~Pukelsheim, {\em Optimal design of experiments} (SIAM, 2006).

\bibitem{helstrom}
C.~W. Helstrom, {\em Quantum detection and estimation theory} (Academic press,
  1976).

\bibitem{holevo}
A.~S. Holevo, {\em Probabilistic and statistical aspects of quantum theory}
  (Edizioni della Normale, 2011).

\bibitem{QSEbook}
M.~G.~A. Paris and J.~E. \v{R}eh\'a\v{c}ek, {\em Quantum State Estimation}
  (Springer, 2004).

\bibitem{hayashi2016quantum}
M.~Hayashi, {\em Quantum Information Theory: Mathematical Foundation}
  (Springer, 2016).

\bibitem{petz}
D.~Petz, {\em Quantum information theory and quantum statistics} (Springer
  Science \& Business Media, 2007).

\bibitem{teo2016}
Y.~S. Teo, {\em Introduction to quantum-state estimation} (World Scientific,
  2016).

\bibitem{kosut2004optimal}
R.~Kosut, I.~A. Walmsley and H.~Rabitz, Optimal experiment design for quantum
  state and process tomography and hamiltonian parameter estimation  (2004).

\bibitem{nunn2010optimal}
J.~Nunn, B.~J. Smith, G.~Puentes, I.~A. Walmsley and J.~S. Lundeen, {\em
  Physical Review A} {\bf 81} (Apr 2010) p. 042109.

\bibitem{ballo2012convex}
G.~Ball{\'o}, K.~M. Hangos and D.~Petz, {\em IEEE transactions on automatic
  control} {\bf 57}  (2012) 2056.

\bibitem{ruppert2012optimal}
L.~Ruppert, D.~Virosztek and K.~Hangos, {\em Journal of Physics A: Mathematical
  and Theoretical} {\bf 45}  (2012) p. 265305.

\bibitem{stm12}
T.~Sugiyama, P.~S. Turner and M.~Murao, {\em Physical Review A} {\bf 85}
  (2012) p. 052107.

\bibitem{gns19}
Y.~Gazit, H.~K. Ng and J.~Suzuki, {\em Physical Review A} {\bf 100} (Jul 2019)
  p. 012350.

\bibitem{fh97}
V.~V. Fedorov and P.~Hackl, {\em Model-oriented design of experiments}
  (Springer Science \& Business Media, 2012).

\bibitem{pp13}
L.~Pronzato and A.~P{\'a}zman, {\em Design of experiments in nonlinear models}
  (Springer \& Business Media, 2013).

\bibitem{fl14}
V.~V. Fedorov and S.~L. Leonov, {\em Optimal design for nonlinear response
  models} (CRC Press, 2013).

\bibitem{chaloner1995bayesian}
K.~Chaloner and I.~Verdinelli, {\em Statistical Science}   (1995) 273.

\bibitem{dasgupta199629}
A.~DasGupta, {\em Handbook of Statistics} {\bf 13}  (1996) 1099.

\bibitem{ryan2016review}
E.~G. Ryan, C.~C. Drovandi, J.~M. McGree and A.~N. Pettitt, {\em International
  Statistical Review} {\bf 84}  (2016) 128.

\bibitem{lu2020generalized}
X.-M. Lu, Z.~Ma and C.~Zhang, {\em Physical Review A} {\bf 101} (Feb 2020) p.
  022303.

\bibitem{yamagata11}
K.~Yamagata, {\em International Journal of Quantum Information} {\bf 9}  (2011)
  1167.

\bibitem{zhu2015information}
H.~Zhu, {\em Scientific reports} {\bf 5}  (2015) 1.

\bibitem{kw60}
J.~Kiefer and J.~Wolfowitz, {\em Canadian Journal of Mathematics} {\bf 12}
  (1960) 363.

\bibitem{nagaoka89-2}
H.~Nagaoka, On the parameter estimation problem for quantum statistical models,
  in {\em Asymptotic Theory of Quantum Statistical Inference: Selected
  Papers\/},  ed. M.~Hayashi (World Scientific, 2005).

\bibitem{HM98}
M.~Hayashi and K.~Matsumoto, Statistical model with measurement degree of
  freedom and quantum physics, in {\em Surikaiseki Kenkyusho Kokyuroku\/},
  1998.
\newblock (English translation available in \cite{hayashi}).

\bibitem{BNG00}
O.~Barndorff-Nielsen and R.~Gill, {\em Journal of Physics A: Mathematical and
  General} {\bf 33}  (2000) p. 4481.

\bibitem{at95}
M.~Akahira and K.~Takeuchi, {\em Non-regular statistical estimation} (Springer
  Science \& Business Media, 2012).

\bibitem{sp2020_ijqi}
L.~Seveso and M.~G.~A. Paris, {\em International Journal of Quantum
  Information} {\bf 18}  (2020) p. 2030001.

\bibitem{D_Ariano_2005}
G.~M. D'Ariano, P.~L. Presti and P.~Perinotti, {\em Journal of Physics A:
  Mathematical and General} {\bf 38} (jun 2005) p. 5979.

\bibitem{fujiwara06}
A.~Fujiwara, {\em Journal of Physics A: Mathematical and General} {\bf 39}
  (2006) p. 12489.
  
  \bibitem{HM08}
M.~Hayashi and K.~Matsumoto, {\em Journal of Mathematical Physics} {\bf 49}
  (2008) p. 102101.

\bibitem{KG09}
J.~Kahn and M.~Guta, {\em Communications in Mathematical Physics}
  {\bf 289}  (2009) 597.

\bibitem{YFG13}
K.~Yamagata, A.~Fujiwara and R.~D. Gill, {\em The Annals of Statistics} {\bf
  41}  (2013) 2197.

\bibitem{YCH18}
Y.~Yang, G.~Chiribella and M.~Hayashi, {\em Communications in Mathematical
  Physics} {\bf 368}  (2019) 223.

\bibitem{young}
T.~Y. Young, {\em Information Sciences} {\bf 9}  (1975) 25.

\bibitem{nagaoka87}
H.~Nagaoka, On fisher information of quantum statistical models, in {\em
  Asymptotic Theory of Quantum Statistical Inference: Selected Papers\/},  ed.
  M.~Hayashi (World Scientific, 2005).

\bibitem{bc94}
S.~L. Braunstein and C.~M. Caves, {\em Physical Review Letters} {\bf 72} (May
  1994) 3439.

\bibitem{ANbook}
S.-I. Amari and H.~Nagaoka, {\em Methods of information geometry} (American
  Mathematical Soc., 2007).

\bibitem{suzuki2020quantum}
J.~Suzuki, Y.~Yang and M.~Hayashi, {\em Journal of Physics A: Mathematical and
  Theoretical} {\bf 53}  (2020) p. 453001.

\bibitem{gm00}
R.~D. Gill and S.~Massar, {\em Physical Review A} {\bf 61} (Mar 2000) p.
  042312.

\bibitem{nagaoka89}
H.~Nagaoka, {\em IEICE Tech Report} {\bf IT 89-42}  (1989) 9.
\newblock (Reprinted in \cite{hayashi}).

\bibitem{hayashi97}
M.~Hayashi, A linear programming approach to attainable cramer-rao type bound,
  in {\em Quantum Communication, Computing, and Measurement\/},  eds.
  O.~Hirota, A.~S. Holevo and C.~M. Caves (Plenum, New York, 1997).

\bibitem{bhatia2013matrix}
R.~Bhatia, {\em Matrix analysis} (Springer Science \& Business Media, 2013).

\bibitem{PhysRevA.90.012115}
H.~Zhu, {\em Physical Review A} {\bf 90} (Jul 2014) p. 012115.

\bibitem{LFGKC}
N.~Li, C.~Ferrie, J.~A. Gross, A.~Kalev and C.~M. Caves, {\em Physical Review
  Letters} {\bf 116}  (2016) p. 180402.

\bibitem{ZH18}
H.~Zhu and M.~Hayashi, {\em Physical Review Letters} {\bf 120}  (2018) p.
  030404.

\bibitem{suzuki2020nuisance}
J.~Suzuki, {\em Journal of Physics A: Mathematical and Theoretical} {\bf 53}
  (2020) p. 264001.

\bibitem{tsang2020semipara}
M.~Tsang, F.~Albarelli and A.~Datta, {\em Physical Review X} {\bf 10} (Jul
  2020) p. 031023.

\bibitem{hayashi}
H.~Masahito (ed.), {\em Asymptotic theory of quantum statistical inference:
  selected papers} (World Scientific, 2005).

\end{thebibliography}
\end{document}